\documentclass[a4paper,cleveref,USenglish]{lipics-v2021}

\usepackage{xcolor}
\usepackage{listings}
\usepackage{verbatim}

\usepackage{tikz}
\usepackage{amsmath}
\usepackage{amssymb}
\usepackage{mathtools}

\usepackage{enumitem}
\crefname{enumi}{step}{steps}
\crefname{property}{property}{properties}

\nolinenumbers

\newcounter{gnote_counter}

\newcounter{enote_counter}

\newcommand{\codestyle}[1]{{\ttfamily #1}}

\newcommand{\term}[1]{{\em #1}}
\newcommand{\illus}[1]{#1}

\newcommand{\ins}{\codestyle{insert}}
\newcommand{\del}{\codestyle{delete}}
\newcommand{\contains}{\codestyle{contains}}
\newcommand{\size}{\codestyle{size}}
\newcommand{\countt}{\codestyle{count}}
\newcommand{\summ}{\codestyle{sum}}
\newcommand{\minn}{\codestyle{min}}
\newcommand{\maxx}{\codestyle{max}}
\newcommand{\select}{\codestyle{select}}
\newcommand{\rank}{\codestyle{rank}}

\newcommand{\cas}{\texttt{CAS}}

\newcommand{\fastquer}{\texttt{FastQueryTree}}
\newcommand{\fastup}{\texttt{FastUpdateTree}}

\DeclareMathOperator*{\Oplus}{\oplus}

\newtheorem{property}[theorem]{Property}

\title{Concurrent aggregate queries}

\author{Gal Sela}{Technion, Israel}{galy@cs.technion.ac.il}{0000-0003-2342-6955}{}
\author{Erez Petrank}{Technion, Israel}{erez@cs.technion.ac.il}{0000-0002-6353-956X}{}

\authorrunning{G.\, Sela and E.\, Petrank}

\Copyright{Gal Sela and Erez Petrank}



\begin{document}
\newcommand{\enote}[1]{{\color{blue}\begin{quote}{\bf Erez's note:} #1\end{quote}}}
\newcommand{\gnote}[1]{{\color{teal}\begin{quote}{\bf Galy's note:} #1\end{quote}}}
\newcommand{\ignore}[1]{}

\maketitle

\begin{abstract}
Concurrent data structures serve as fundamental building blocks for concurrent computing. Many concurrent counterparts have been designed for basic sequential mechanisms; however, one notable omission is a concurrent tree that supports aggregate queries. Aggregate queries essentially compile succinct information about a range of data items, for example, calculating the average salary of employees in their 30s. Such queries play an essential role in various applications and are commonly taught in undergraduate data structures courses.
In this paper, we formalize a type of aggregate queries that can be efficiently supported by concurrent trees and present a  design for implementing these queries on concurrent trees. We bring two algorithms implementing this design, where one optimizes for tree update time, while the other optimizes for aggregate query time. We analyze their correctness and complexity, demonstrating the trade-offs between query time and update time.

\end{abstract}

\section{Introduction}\label{section:intro}
Concurrent programs rely on concurrent data structures as a foundational component.
Considerable effort has been dedicated to constructing efficient concurrent data structures that allow using data structures in a concurrent setting. However, not all sequential functionalities have been extended to the concurrent setting. In this paper we look at such a functionality whose concurrent version has not been addressed: efficient aggregate queries.
An aggregate query is a query whose answer aggregates into a succinct value information about a range of elements with consecutive keys in the data structure. For instance, a data structure holding employee records sorted by age may be queried regarding the average salary of employees in a certain age range. 
Efficient aggregate queries in this spirit were not designed for concurrent data structures.

It is important to build efficient concurrent algorithms for aggregate queries, as sequential aggregate queries are used in various applications, and a concurrent extension may scale their execution on a multi-core machine.
For instance, order-statistic trees \cite{cormen2022introduction}, which support the \select{}$(i)$ and \rank{}$(key)$ aggregate queries (returning the element with the $i$-th smallest key, and the position of $key$, respectively), are used in Python libraries for sorted containers \cite{pythonBlist,pythonSortedcontainers} to efficiently support the basic operations of accessing $collection[i]$ and querying $collection.index(key)$ respectively.
In C++, Boost library for example offers an order-statistic tree
\cite{boostRankedIndices}. 
Furthermore, various text editors use augmented tree structures similar to order-statistic trees to efficiently access and manipulate text at a certain location indicated by the editor's cursor. One example is the \term{rope} data structure \cite{boehm1995ropes}, which has various implementations (e.g., \cite{ropey,cropRope,ropeSGI}), and underlies the buffer implementation of some text editors like Xi \cite{xiEditor} and Lapce \cite{lapce}.

Naively, one could answer an aggregate query on a sequential data structure by traversing the relevant elements. The concurrent counterpart would be taking a linearizable snapshot of the data structure and traversing it. Previous works on range queries accomplished that \cite{wei2021constant,petrank2013lock,nelson2022bundling,arbel2018harnessing}, but the traversal in this approach costs time linear in the number of elements in the queried range, which is highly inefficient for aggregate queries that may be answered using some metadata without traversing all the relevant elements.
There has been work on implementing specific concurrent aggregate queries more efficiently: \cite{sela2021concurrentSize} proposed a way to efficiently support a \size{} query returning the total number of elements in a concurrent data structure. They hold appropriate central metadata regarding the data structure's size. This metadata is updated once on behalf of each effectual operation (we call an operation that modifies the data structure, like a successful \ins{} or \del{}, an \term{effectual operation}). \size{} queries take a snapshot of the metadata without accessing the data structure's elements themselves. 
However, their mechanism is not extensible to our case, where we want to efficiently answer a query about a specific range, given as an input to the query. To this end, each effectual operation will update multiple metadata pieces across multiple nodes, and it should carry out all the updates (of its target element in the data structure and the multiple metadata values) seemingly-atomically so that queries will get a coherent view of the data structure's state.

More specifically, we focus our attention on aggregate queries on trees. We will look at external binary search trees (where external means they hold the elements in the leaves) though our work could be extended to other trees as well. For efficiently answering aggregate queries on sequential trees, one could place in each tree node suitable metadata that is a function of the elements in the leaves of the node's subtree. For instance, an order-statistic tree augments a binary search tree with a size field expressing the number of elements in the node's subtree. The metadata function should be chosen to be one that effectual operations could maintain during their root-to-target-leaf traversal for not harming their asymptotic time complexity, and also one that aggregate queries could use to get an answer via root-to-leaf traversals (instead of naively traversing the relevant elements), thus executing in time linear in the traversed path length instead of at the number of elements in the query's range.
We formalize the type of addressed aggregate functions and queries in \Cref{section: aggregate}.

Extending such augmented trees to support concurrency is not simple: while each effectual operation affects multiple locations (its target-leaf area and metadata fields in the nodes along its root-to-leaf path), a query should somehow obtain a consistent view of the nodes it traverses, including their metadata fields.
If insertions and deletions simply update the relevant metadata fields one by one, a query might obtain an inconsistent view of the metadata, in which ongoing effectual operations might be reflected only in part of the obtained fields. The query should obtain a linearizable snapshot of the query path including its metadata, where each effectual operation is either fully absent or fully reflected in all relevant snapshotted fields.

To get a snapshot while viewing the multi-point updates per insertion or deletion as atomic, one could naively use a lock, or employ transactions---for atomically updating the target-leaf area and the metadata along the root-to-leaf path by each successful insertion or deletion, as well as for reading the root-to-leaf path by queries. But that would be inefficient since all effectual operations as well as concurrent queries would be serialized one after another because they all synchronize on the metadata in the root node.
We wish to take an alternative more fine-grained approach, where we let queries traverse the tree without taking a lock, but make sure to fill in the missing updates in the otherwise potentially partial picture a query obtains.
We employ two mechanisms to achieve this goal: multi-versioning, and announcements of ongoing effectual operations. The first enables queries to ignore effectual operations that are considered to occur after them, and the second enables them to take into account all effects of  effectual operations that are considered to occur before them.

In more detail, each effectual operation and each aggregate query get a timestamp. An aggregate query with timestamp $ts$ is linearized (i.e., considered to occur) after all effectual operations with timestamp $\leq ts$ and before all effectual operations with timestamp $> ts$. 
By multi-versioning we mean that effectual operations which update values in a timestamp later than an ongoing query's timestamp are responsible to leave old versions of the updated objects. These old versions enable queries to obtain a snapshot of the relevant parts of the data structure (including metadata values) without taking into account new updates that are considered to happen after them:
each query grabs a timestamp and then builds its view of the query path by reading object versions tagged with this timestamp (to be precise, with the biggest timestamp that is $\leq$ this timestamp). 
Indeed, such a snapshot will not reflect updates with a timestamp greater than the query's timestamp; however, the snapshot might contain partial updates for some still-ongoing insertions or deletions considered to precede the query according to their timestamp.
Therefore, insertions and deletions globally announce themselves, and queries read these announcements to fill in missing details about them by themselves, and form the desired full view of their traversed path. 

Different approaches could be taken toward the operations announcement and the aggregate metadata representation, optimizing for the time complexity of either effectual operations or aggregate queries. We bring two algorithms implementing our design: \fastup{} optimizing for tree update time, in fact incurring zero additional asymptotic time on the original tree operations, and \fastquer{} optimizing for aggregate query time, reducing its time to be linear in its traversal length times a factor dependent on the number of concurrent operations (in comparison to time linear in the number of elements in the queried range in the naive implementation).
Applications in which it is important that the original tree operations run fast should use \fastup{}, while applications that require fast aggregate queries should use \fastquer{}.

To reduce the contention on effectual operations incurred by our design, \fastup{} lets them work mostly on single-writer fields written only by the thread that performs the operation. 
Both the object in which threads announce their effectual operations, and the aggregate metadata field in each tree node, are arrays with a single-writer cell per thread. This demonstrates a time-space trade-off: effectual operations can execute faster by paying in more space. The trade-off between the original tree operations' time and aggregate query time is also apparent here, as aggregate queries will have to invest more work to calculate aggregate values based on the per-thread aggregate metadata array. 

\fastup{} demonstrates that effectual operations do not have to serialize (order themselves one in respect to another) and they in fact do not need to be aware of other effectual operations at all, which is interesting to note as it might initially seem like they must be ordered for the metadata to be updated of relevant effectual operations by their order.
On the other hand, in \fastquer{}, for aggregate queries to run faster, they are not responsible for combining per-thread metadata into a unified aggregate value, and instead effectual operations serialize in order to know which effectual operations precede them and help update metadata on behalf them. The aggregate metadata is made of a single field (and not an array), updated on behalf of relevant effectual operations by their order.
To serialize themselves, effectual operations announce themselves by enqueueing their announcement to a global queue.

As parallelization of aggregate queries is notably absent in the literature, it is perhaps unsurprising that this challenge has attracted attention from other researchers. In an independent and very interesting complementary research, the problem of concurrent aggregate queries is also being explored~\cite{kokorin2023wait}, where a different approach is presented.
In~\cite{kokorin2023wait} all operations are first serialized on a queue of operations at the root. Each operation then helps all preceding operations in the queue to advance to the appropriate child before proceeding to the next node in its own traversal. Queues are used in each node of the tree to hold all operations that still need to be executed on that node's subtree. Reminiscent of hand-over-hand synchronization, this approach provides a virtual snapshot for the queries, as each operation views the tree with all previous operations have already been executed and none of the subsequent operations have. 
Their alternative approach complements ours, providing a more comprehensive view of possible efficient solutions for aggregate queries. Their solution optimizes aggregate query performance at the expense of the other operations, while our approach also includes the \fastup{} solution, which offers optimal complexity for the original tree operations. 

\smallskip

\noindent {\bf Contribution.} We present a formalization of aggregate queries that may be efficiently supported on concurrent trees, and their related aggregate metadata functions.
We then present a design for supporting such queries on concurrent trees, with two different implementations, presenting a trade-off between the time complexity of the original tree operations and the newly-added aggregate queries.
\smallskip 

\noindent{\bf Organization.} The formalization of aggregate queries appears in \Cref{section: aggregate}. Some terminology required to discuss our algorithms is brought in \Cref{section: terminology}, followed by our design presentation in \Cref{section: algorithms}
and the analysis of the algorithms in \Cref{section: analysis}. Related work is covered in \Cref{section: related work}. We conclude with a discussion of future directions in \Cref{section: discussion}.

\section{Aggregate metadata and aggregate queries}\label{section: aggregate}

We look at aggregate queries on binary trees, using metadata placed in each node aggregating information about its subtree. The basic idea is to use this metadata to answer queries efficiently without traversing all the elements in the query's range, while making sure not to substantially harm the asymptotic time complexity of insertions and deletions, which should be able to maintain the metadata throughout their root-to-leaf traversal, as the only affected metadata should lie along the path to their target key.

\subsection{Aggregate functions used for metadata}

The aggregate metadata we will add to tree nodes is the value of an aggregate function $f$ applied to the set of $(key,value)$ elements in the leaves of the node's subtree. An aggregate function is a function $f:\mathcal{P}(A)\setminus\{\phi\} \to B$, where $A,B$ are non-empty sets.
Our aggregate functions' domain would be all the non-empty subsets of the set of possible $(key,value)$ elements in the tree's leaves (denoted by $A$).
We note that aggregate functions are usually referred to as applied to multisets, since database rows may include repetitions. However, since we treat the domain as pairs of $(key,value)$ and the tree's keys as unique, we define aggregate functions to operate on sets. 
The codomain of aggregate functions, denoted here by $B$, could be for example $\mathbb{R}$, $\mathbb{Z}_2$, the set of possible tree keys, the set of possible tree values, etc., or a product of several such sets (so that $B$'s elements are tuples).
Next we bring a definition that will give us the first useful property we need of aggregate functions:

\begin{definition}[additive aggregate function]\label{definition: additive aggregate function}
We say that an aggregate function $f:\mathcal{P}(A)\setminus\{\phi\} \to B$ is \term{additive} if there exists a binary operation $\oplus:B\times B\to B$ such that for every disjoint $X_1, X_2 \in A$,
\begin{equation}\label{equation: additive aggregate function}
f(X_1\uplus X_2)=f(X_1)\oplus f(X_2)
\end{equation}
\end{definition}

To better understand how an additive aggregate function generally looks, we bring the following equivalent definition:

\begin{definition}[additive aggregate function -- alternative]\label{alternative definition: additive aggregate function}
An aggregate function $f:\mathcal{P}(A)\setminus\{\phi\} \to B$ is \term{additive} if there exists a binary operation $\oplus:B\times B\to B$ such that $(B,\oplus)$ is a commutative semigroup (namely, $\oplus$ is associative and commutative) and $f$ satisfies: 
\[
f(X)=\Oplus_{a\in X}f(\{a\})
\]
(where $\oplus$, though binary, may be extended to be applied to any number of elements from $A$ thanks to the operation's associativity).
\end{definition}

In fact, we could build $f$ on top of any ``base'' function $F$: for any commutative semigroup $(B,\oplus)$ and any $F:A\to B$, $f(X)\coloneq \Oplus_{a\in X}F(a)$ is an additive aggregate function.

Examples of useful additive aggregate functions include \size{} (sometimes denoted by \countt{}), for which $B\coloneq\mathbb{Z}$ and $\oplus \coloneq +$ (simple addition), and \summ{} over the keys or over the values with $\oplus$ taken to be addition over the appropriate domain, e.g. $\mathbb{R}$ if the keys or values are taken from $\mathbb{R}$. The sum of squared values $f(X)=\sum_{(key,value)\in X}value^2$ may be useful for sample variance calculation. Product is another associative commutative operation that could be taken as $\oplus$ to produce the multiplication of the keys or values. 

Requiring the metadata in tree nodes to be a value of an additive aggregate function over the set of $(key,value)$ elements in the leaves of the node's subtree ensures that the metadata in each node may be directly updated upon an insertion to its subtree:

\begin{property}\label{property: update for additive}
Upon an insertion of a $(key,value)$ element into a certain subtree, the new updated value of the metadata in the subtree's root may be calculated by $old\oplus f(\{(key,value)\})$ where $old$ is the old metadata value.
\end{property}

\minn{} and \maxx{} are also additive aggregate functions, where $\oplus$ is taken to be \minn{} or \maxx{} respectively, but our work will not support using them as metadata in nodes since they do not satisfy the following property:

\begin{definition}[subtractive aggregate function]\label{definition: subtractive aggregate function}
We say that an aggregate function $f:\mathcal{P}(A)\setminus\{\phi\} \to B$ is \term{subtractive} if it is additive, and there exists a subtractive binary operation $\ominus:B\times B\to B$ such that for every disjoint $X_1, X_2 \in A$,
\begin{equation}\label{equation: subtractive aggregate function}
f(X_2)=f(X_1\uplus X_2)\ominus f(X_1)    
\end{equation}
\end{definition}

The following definition is equivalent to \Cref{definition: subtractive aggregate function}:

\begin{definition}[subtractive aggregate function -- alternative]\label{alternative definition: subtractive aggregate function}
An aggregate function $f:\mathcal{P}(A)\setminus\{\phi\} \to B$ is \term{subtractive} if it is additive with an operation $\oplus:B\times B\to B$, and $(B,\oplus)$ is a group (namely, $\oplus$ has an identity element and every element of $B$ has an inverse element).
\end{definition}

This alternative definition helps to materialize the $\ominus$ operator from \Cref{definition: subtractive aggregate function}: as each $b_2\in B$ has an inverse element $-b_2$, we define $\ominus$ as follows (and this satisfies the requirement of \Cref{equation: subtractive aggregate function}):
\[
b_1\ominus b_2\coloneq b_1\oplus -b_2
\]

All additive aggregate function examples brought above, except for \minn{} and \maxx{}, are subtractive. For all of them $\ominus$ should be subtraction, other than for product for which it should be division.

The metadata field we will add to the tree nodes will hold the value of an aggregate function which is not only additive but rather also subtractive. 
This way, the following property will hold and enable to update the metadata accordingly:

\begin{property}\label{property: update for subtractive}
Upon a deletion of a $(key,value)$ element from a certain subtree, the new updated value of the metadata in the subtree's root may be calculated by $old\ominus f(\{(key,value)\})$ where $old$ is the old metadata value.
\end{property}

On deletion we conceptually invert the node's metadata value to its state without the deleted item, which we could not have done for non-subtractive aggregate functions like \minn{}.
Being able to directly update the metadata in a certain node to reflect a deletion in its subtree, without re-calculating the metadata node by node from the location of the deleted leaf upwards, is especially important for our algorithms, where the deletion initiator is not the only who needs to calculate its effect on the metadata in the deleted leaf's ancestors. Other operations whose root-to-leaf path intersects the deletion's path might need to do so as well in ancestors mutual with this deletion, and they should not traverse all the way from the deleted leaf to the relevant ancestor which might be costly. 

\subsection{Aggregate queries}

An aggregate query on a data structure returns a result based on multiple data elements of the data structure. We look specifically at trees, whose augmentation with appropriate metadata in all nodes may enable aggregate queries to execute efficiently through root-to-leaf traversals.
Metadata obtained during the traversals may be used both to navigating through the tree, choosing the appropriate path to traverse, and for calculating the query's result.
Some queries require multiple root-to-leaf traversals for computing their answer. These traversals may be independent of each other, which means they could be executed concurrently, followed by a central calculation of the query's answer using their results. But there are also queries  that require a serial execution of traversals, which is the case when each traversal depends on the result of the previous traversal.
Accordingly, we next define a \term{simple} aggregate query, which executes only independent traversals; a general aggregate query is a chain of one or more simple aggregate queries composed one on another: the user's input is the input of the first simple query in the chain, the output of the $i$-th simple query in the chain is the input of the $(i+1)$-st query, and the output of the last simple query is the output of the whole query.
\Cref{section: aggregate examples} brings examples of supported aggregate queries, and gives intuition to what kind of queries require each part of the construction: a basic traversal, multiple traversals, and chained traversals.

A simple aggregate query performs one or more independent root-to-leaf traversals to gather the information required to answer the query, and then computes the answer using the traversals' results. The traversals may be executed concurrently as they are independent of each other. 
All traversals have a fixed structure that appears in the template in \Cref{fig: basicAggregateQuery template} detailed below; the only difference between them is in the \codestyle{shouldDescendRight} method called for each node to determine to which child the traversal should proceed, and hence the leaf they eventually arrive at. A traversal outputs a tuple (\textit{aggValue, leaf}), where \textit{aggValue} is the value of the metadata subtractive aggregate function on the set of $(key,value)$ pairs of all leaves in the key range $(-\infty, k)$ with $k$ being the key in the leaf ending the traversal, and \textit{leaf} is this leaf's object. 
A simple aggregate query takes an input (e.g., a key or an index) and needs to return the required output, namely, the query's answer. Its definition is made of two components: as many \codestyle{shouldDescendRight} methods as the traversals it needs (each of them may integrally use the query's input), and a \codestyle{computeAnswer} method that takes the list of the traversals' outputs and computes the query's answer.

\begin{figure}
\begin{lstlisting}
@\underline{aggregateTraversalTemplate<shouldDescendRight>()}@:
@\textit{init:}@
    node = tree.root
    skippedNodesAggValue = identityElement$_{(B,\oplus_f)}$
@\textit{traverse:}@
    while node is not a leaf:
        aggValueUpToCurrentKey = skippedNodesAggValue $\oplus_f$ node.left.aggValue@\label{line:calc agg}@
        if shouldDescendRight(aggValueUpToCurrentKey, node.key):@\label{line:shouldDescendRight}@
            node = node.right
            skippedNodesAggValue = aggValueUpToCurrentKey@\label{line:update skippedNodesAggValue}@
        else:
            node = node.left
return (skippedNodesAggValue, node)
\end{lstlisting}
\caption{Template for a basic aggregate query}\label{fig: basicAggregateQuery template}
\end{figure}

Next, we present the traversal algorithm ran during aggregate queries execution. 
Throughout the paper, we denote the nodes' metadata type by $B$ and the metadata subtractive aggregate function (whose codomain is $B$) by $f$.
We further denote $f$'s corresponding operators (from \Cref{definition: additive aggregate function,definition: subtractive aggregate function}) by $\oplus_f$ and $\ominus_f$, and the identity element of the group $(B,\oplus_f)$ by \codestyle{identityElement$_{(B,\oplus_f)}$}.
The general traversal template, shown in \Cref{fig: basicAggregateQuery template}, takes as parameter a \codestyle{shouldDescendRight} method (defined by the aggregate query).
It performs a root-to-leaf traversal on the tree. 
At any point in the traversal, \codestyle{skippedNodesAggValue} holds the value of $f$ on the set of $(key,value)$ pairs of all leaves found in subtrees that the traversal has jumped over so far (namely, descended to the right while they were in the left subtree). At the end of the traversal, this will be the set corresponding to all keys preceding the key of the leaf the traversal has reached. 
For each traversed node, the aggregate value \codestyle{skippedNodesAggValue} computed so far is combined using $\oplus_f$ with the aggregate value of the current left subtree, found in the metadata of the left child, to form \codestyle{aggValueUpToCurrentKey}---which represents the value of $f$ on the set of $(key,value)$ pairs of all leaves with $key\leq$ the key of the current node in the traversal (\Cref{line:calc agg}).
The computation of this value is made possible using one simple $\oplus$ operation thanks to using an additive aggregate function on the subtree's leaves as the node's metadata.
Then the query-specific method \codestyle{shouldDescendRight} is called to determine to which child the traversal should proceed (\Cref{line:shouldDescendRight}). It takes \codestyle{aggValueUpToCurrentKey} and the current node's key. 
(For example, a traversal that aggregates $f$ for all keys up to a certain key $k$ should search for $k$, and the query would accordingly define a \codestyle{shouldDescendRight} method that returns true iff $k \geq$ the current node's key---we assume a binary search tree where keys equal to or greater than a node's key are stored in its right subtree. An example to \codestyle{shouldDescendRight} that requires \codestyle{aggValueUpToCurrentKey} to make its decision is a \select{} query, for which $f$ counts the number of elements, and whose traversal proceeds to the right child if the required leaf's index $>$ \codestyle{aggValueUpToCurrentKey}.)
In case of descending to the right, \codestyle{skippedNodesAggValue} is updated to take into account the leaves of the current left subtree (\Cref{line:update skippedNodesAggValue}).
The traversal stops when it reaches a leaf node, and returns \codestyle{skippedNodesAggValue} and the leaf object.

\section{Terminology}\label{section: terminology}

We will look at binary search trees implementing dictionaries.
A \term{dictionary} (synonymously \term{map} or \term{key-value map}) is a collection of distinct keys with associated values, supplying the following interface operations: an \ins($k$, $v$) operation which inserts the key $k$ with the associated value $v$ if the key does not exist or else returns a failure; a \del($k$) operation which deletes $k$ and its value if $k$ exists and returns the value or else returns a failure; and a \contains($k$) operation which returns $k$'s value if $k$ exists else returns NOT\_FOUND.
We call \ins{} and \del{} that return failure \term{failing}, otherwise they are \term{successful}. We call successful \ins{} and \del{} \term{effectual operations}.
We will assume binary search trees where keys smaller than a node's key are stored in its left subtree and keys equal to or greater than a node's key are stored in its right subtree.
We will look specifically at external trees, i.e., their items are found in the leaves.

We assume the basic asynchronous shared memory model \cite{herlihy1991wait}, in which a fixed set of threads communicate through memory access operations.
An execution on a concurrent data structure is considered \term{linearizable} \cite{herlihy1990linearizability} if each method call appears to take effect at once, between its invocation and its response events, at a point in time denoted its \term{linearization point}, in a way that satisfies the sequential specification of the objects.
A concurrent data-structure is \term{linearizable} if all its executions are linearizable.
The two tree algorithms we present, \fastup{} and \fastquer{}, are linearizable.

\section{The design}\label{section: algorithms}

An overview of our design for augmenting a concurrent binary search tree with subtractive aggregate metadata to support aggregate queries appears in \Cref{subsection: design overview}. We describe the two different approaches towards implementing this design in \Cref{subsection: the 2 algs}. We then delve into details---of the tree on which we demonstrate our methodology in \Cref{subsection: base tree}, the common backbone of the two algorithms in \Cref{subsection: design backbone}, the unique details of each algorithm in \Cref{section: fastup,section: fastquer} and optimizations in \Cref{section: optimizations}.

\subsection{Design overview}\label{subsection: design overview}
We wish to extend a basic tree with support for efficient aggregate queries. For that we add to tree nodes aggregate metadata that will enable to answer them efficiently. The aggregate metadata in each tree node equals to the value of a subtractive aggregate function $f$ applied to the set of $(key,value)$ elements in the leaves of the node's subtree.

We need to correctly answer aggregate queries even if they are concurrent with operations that update the tree. The challenge is to overcome the fact that effectual operations carry out multiple modifications of the tree, and let aggregate queries observe a consistent view of the parts they traverse in the data structure, as if each concurrent effectual operation has completely taken place or did not start at all. 
Each effectual operation and each aggregate query obtain a timestamp. Every query should observe all modifications related to effectual operations with timestamps $\leq$ its timestamp, and not see modifications related to effectual operations with a greater timestamp.

For that, on the one hand we need a query to consider all modifications by effectual operations, which run concurrently with the query and have a timestamp $\leq$ its timestamp, even if some of these modifications have not yet occurred. To this end, ongoing effectual operations announce themselves by adding an \textit{Update} object with their details to a global \textit{CurrUpdates} object, to enable concurrent queries to complete the missing details by themselves. Among other details, the \textit{Update} object contains a $timestamp$ field indicating the operation's timestamp, and a $done$ flag indicating whether the operation is done both updating affected aggregate fields and applying itself to the tree. (All the fields are detailed in \Cref{subsection: node and Update fields}.)

On the other hand, we also need to prevent effectual operations, which run concurrently with a query and have a greater timestamp than the query's timestamp, from overriding data the query is about to use with new data. To this end, we employ versioning for modifiable fields in the tree's nodes: effectual operations leave old versions of the data for the queries to inspect, and write the new values in new versions they create for the relevant fields. 
More specifically, we use timestamped version lists for both the child pointers and the added aggregate metadata field in the tree nodes. These versioned fields are made of a linked list of values tagged with descending timestamps. Reading them may be performed with or without an input timestamp, while writing to them must be done with an input timestamp, as detailed next.

A \term{versioned read} takes a timestamp $ts$, and traverses the list until reaching a version with timestamp $\leq ts$, whose value it returns. A \term{standard read} returns the value in the first (most recent) version in the list. A \term{standard timestamped read} returns the value in the first version in the list and its timestamp. A write to a versioned field (unprotected write, which should be performed while it is guaranteed that no other thread concurrently tries to write to the field) takes a value and a timestamp $ts$, and links a new version to the head of the linked list of versions with the new value and timestamp $ts$. A thread-safe write to a versioned field (namely, while other threads might concurrently try to write to it) takes a value and a timestamp \textit{newTs}, in addition to \textit{lastTs}---the timestamp expected to be the most recent one in the list; it obtains the current first version, and if its timestamp is \textit{lastTs} it tries to link before it (to the head of the linked list) a new version with the new value and timestamp \textit{newTs} using a compare-and-swap (\cas{}).
The full pseudocode for versioned fields appears in \Cref{subsection: versioned field}. Which kind of read or write is used in what scenarios will become clear in the following sections.

\subsection{The two algorithms}\label{subsection: the 2 algs}

The two proposed algorithms share the same backbone, but they handle differently the way operations obtain a timestamp and announce themselves in \textit{CurrUpdates} in case of effectual operations, as well as the aggregate metadata representation. 
These could be biased in favor of the time complexity of either effectual operations or aggregate queries.
We design two algorithms to handle the timestamps and aggregate values---\fastup{} that preserves the asymptotic time complexity of operations in the base tree algorithm, incurring no additional asymptotic time cost on them; and \fastquer{} which offers a better time complexity for aggregate queries.

In the design of the tree algorithm's extension for supporting aggregate queries, effectual operations have to perform several additional steps in which they might potentially contend with operations of other threads: globally announce and unannounce themselves and update the metadata fields affected by the operation. 
\fastup{} aims to reduce the contention on effectual operations incurred by our extension, and thus lets effectual operations work mostly on single-writer fields written only by the thread that performs the operation. This manifests in both the announcement mechanism and the aggregate metadata representation:

In \fastup{}, the \textit{CurrUpdates} object---in which effectual operations announce themselves---is an array with a cell per thread to point to its \textit{Update} object. 
When effectual operations announce themselves in this array, they do not order themselves in respect to each other, and there is no variable they serialize on (like obtaining a unique timestamp).
Aggregate queries are the ones to grab a timestamp while incrementing a global \textit{Timestamp} field using a fetch-and-increment; effectual operations only need to obtain a timestamp bigger than the last query's timestamp, for writing their updates of versioned fields in a newer version, not overriding data the query needs.
For that, an effectual operation first announces itself with an unset timestamp, and then it obtains the global timestamp value and sets it in the announcement’s timestamp field using a \cas{}. A concurrent aggregate query might be ahead of it, obtaining the global timestamp value and \cas{}ing it into the announcement’s timestamp, which is what 
aggregate queries do for all effectual operations with an unset timestamp they encounter in their first traversal in \textit{CurrUpdates}.

As for the aggregate metadata field in each \fastup{} node, it is also an array with a cell per thread, where each cell is a versioned field (namely, holds a linked list of versions with different timestamps) containing metadata regarding operations by the associated thread on the node's subtree.
Aggregate queries can correctly calculate the total aggregate value from the per-thread values using $\oplus_f$ (the aggregate function's binary operation), thanks to $\oplus_f$ being commutative and associative---which is the case as we allow only an additive aggregate function $f$ (whose $\oplus_f$ is commutative and associative by \Cref{alternative definition: additive aggregate function}). 

\fastquer{} on the other hand favors the performance of aggregate queries, hence does not let them gather values from a per-thread metadata array; instead, it allocates a single versioned metadata field in each tree node. 
To update such a field to reflect an effectual operation, it needs to know which effectual operations are ordered before it, in order to update the metadata to reflect all relevant operations that occurred so far.
To this end, all effectual operations serialize by enqueueing their \textit{Update} object to a queue, and while doing so they also get a unique timestamp so that the timestamps induce a total order on all effectual operations (specifically, they enqueue an \textit{Update} object with a timestamp greater by 1 than the timestamp of the preceding \textit{Update} object in the queue). Namely, \textit{CurrUpdates} is a queue containing \textit{Update} objects with consecutive timestamps.
Aggregate queries obtain a timestamp (that determines which effectual operations they take into account) by simply reading the timestamp in the \textit{Update} object in the current last node of the queue (namely, the most recent \textit{Update}), which is considered the current global timestamp. Equipped with this timestamp, they know which version of each aggregate metadata field to obtain and which announced effectual operations they should consider.

In \Cref{subsection: design backbone} we elaborate on the common backbone of both algorithms, and the full details of the different points between the two algorithms are deferred to \Cref{section: fastup,section: fastquer}.

\subsection{The base tree}\label{subsection: base tree}

The proposed methodology for augmenting a concurrent binary tree with subtractive aggregate metadata to support aggregate queries, could be applied to different concurrent trees. We will focus on a specific concurrent tree to demonstrate the methodology---the linearizable binary search tree of \cite{david2015asynchronized,david2014designing}.
This is an external full binary tree (the elements are in the leafs; each internal node has two children). Three sentinel nodes are always part of the tree: a root node with the key $-\infty$, a leaf node (which is the root's left child) with the key $-\infty$, and a leaf with the key $\infty$. In particular, the root never changes. Each internal node has a memory-word-sized field containing two special locks, each used to protect the link to another one of its children. It is possible to acquire one of the locks without affecting the other, or permanently acquire them both in a single atomic operation (the latter has a different signature than when each separate lock is acquired, and is used to permanently lock a parent of a leaf that is about to be removed, since as part of the \del{} operation the parent will also be unlinked and its edges will never be modified again). We will refer to acquiring a node's lock associated with a certain child as locking the edge to the child, and permanently acquiring both locks as a permanent lock of the node. 

A \contains{}($k$) operation is simple and operates like in a sequential algorithm: it searches for $k$ until reaching a leaf node, and returns an answer based on the leaf's key.
An \ins{}($k$) operation wishes to insert a new leaf---$N$---with the key $k$. It first searches for $k$ in the tree. If it finds it, it returns failure. Otherwise, it reaches a leaf $L$ with a key $\neq k$. It locks the edge from $P$---$L$'s parent---to $L$ (or restarts if the attempt to lock failed), modifies this edge to point at a new internal node pointing to $L$ and $N$ as its children (and having the right child's key as its key), and then unlocks the edge.
A \del{}($k$) operation starts with searching for $k$ in the tree. If it does not find it, it returns failure. Otherwise, it reaches a leaf $L$ with the key $k$; let $S$ be its sibling, $P$---its parent, and $G$---its grandparent. The \del{} operation locks the edge from $G$ to $P$ and then permanently locks $P$ (or restarts if the attempt to lock any of them failed), modifies the edge from $G$ to $P$ to point at $S$ (which unlinks both $L$ and $P$ from the tree), and then unlocks the edge.

\subsection{Design backbone details}\label{subsection: design backbone}

We describe the backbone common to our two algorithms.
The tree is initialized with three sentinel nodes like in the base tree, and their aggregate value is set to \codestyle{identityElement$_{(B,\oplus_f)}$}.
Next we describe the general scheme of each operation on the tree.
Effectual operations (successful \ins{} and \del{}) acquire the necessary locks, and then before applying the operation to the tree---they globally announce themselves including obtaining a timestamp, and update affected aggregate metadata. Failing \ins{} and \del{} and \contains{} operate as in the base algorithm, but then in the end verify that no ongoing operation has already announced itself and logically deleted the node they found / inserted a node with the key they have not found. Aggregate queries use the aggregate metadata throughout their traversal like sequential aggregate queries, but they also grab a timestamp in the beginning, and obtain versions of child pointers and of aggregate metadata according to this timestamp and according to announced effectual operations.
Details follow.

\subsubsection{\ins{} and \del{} operations}\label{subsection: ins and del}
An \ins{} and a \del{} of a key \textit{k} are performed as follows:

\begin{enumerate}
\item\label{update: start base alg} Run the base tree algorithm until it is about to return \textbf{failure} or until (including) it \textbf{acquires lock/s}.

In the first case, let \textit{L} be the leaf reached by the traversal, \textit{P} be the parent from which it was reached, and \textit{direction} be left or right according to which child of \textit{P} \textit{L} is. Return failure only in the following cases:
    \begin{enumerate}
    \item If it is an insertion (namely, \textit{k} was found in \textit{L}), call \codestyle{isDeleted}(\textit{L, P}) (see \Cref{subsection: auxiliary methods}) and return failure if it returns false.
    \item If it is a deletion (namely, \textit{k} was not found), call \codestyle{getValueIfInserted}(\textit{L, P, direction, k}) (see \Cref{subsection: auxiliary methods}) and return failure if it returns NOT\_FOUND.
    \end{enumerate}
Otherwise, start \Cref{update: start base alg} over.
\item\label{update: announce} \textbf{Announce}: Add an \textit{Update} object to \textit{CurrUpdates}, including \textbf{obtaining $ts$}.
\item\label{update: aggValues} \textbf{Update aggregate values}: 
    \begin{enumerate}
    \item\label{update: aggValues: gather currUpdates} Gather \textit{CurrUpdates} with timestamp $\leq ts$ into \textit{currUpdates}.
    \item\label{update: aggValues: traverse} Traverse from the root to the leaf, and for each \textit{node} (excluding the leaf):
        \begin{enumerate}
        \item\label{update: aggValues: traverse: update} Update \textit{node.aggValue} by versioned writes, using \textit{currUpdates} (while traversing \textit{currUpdates}, ignore done operations and eliminate them from \textit{currUpdates}).
        \item\label{update: aggValues: traverse: choose next} Obtain next traversal's node (left or right child) using a versioned read with $ts$.
        \item\label{update: aggValues: traverse: eliminate from currUps} Eliminate out-of-range effectual operations (namely, effectual operations on keys < \textit{node.key} if proceeding to the right, or with keys $\geq$ \textit{node.key} if proceeding to the left) from \textit{currUpdates}.
        \end{enumerate}
    \end{enumerate}
\item\label{update: apply} \textbf{Apply} the operation to the target leaf area.
\item\label{update: unannounce} \textbf{Unannounce}: set \textit{Update.done} and remove the \textit{Update} from \textit{CurrUpdates}.
\item\label{update: finalize del} \textbf{Finalize} deletion: modify grandparent's pointer.
\item\label{update: unlock} \textbf{Release} the lock.
\end{enumerate}

In more detail, an \ins{}$(k,v)$ or a \del{}$(k)$ starts by running the base tree algorithm, where all reads of versioned fields are standard reads (i.e., return the value in the first version in the list) (\Cref{update: start base alg} above). It executes until one of the following occurs: The first alternative is that the base algorithm is right about to return failure. Before doing so, it must make sure that it is still correct to return failure in our extension of the algorithm---if the operation is an insertion and $k$ was found, it needs to make sure $k$ is not already considered deleted by an ongoing deletion, and if it is a deletion and $k$ was not found, it must make sure $k$ is not already considered inserted by an ongoing insertion
(see \Cref{subsection: auxiliary methods} for details). Otherwise it restarts, starting \Cref{update: start base alg} from the top.

The other alternative is that the necessary locks are acquired, right before the relevant tree's link is modified in the base algorithm. At this point, the operation is guaranteed to succeed but is not yet linearized (which intuitively means it is not yet considered as occurred, and concurrent operations will not consider it). It will be linearized only when completing the next stage---globally announcing itself by adding an $Update$ object with its details to \textit{CurrUpdates}, including obtaining a timestamp $ts$ (\Cref{update: announce}). 
In addition to the $ts$ and $done$ fields that were earlier mentioned, an $Update$ object also includes a $leaf$ field with the node that is inserted or deleted, as well as information about the tree's edge that is about to be modified: its source node (\textit{edgeSource}), its designated new target node (\textit{edgeTarget}) and its direction (\textit{edgeDirection}). 
The operation proceeds to update the aggregate metadata in the nodes in its root-to-leaf path. When updating, it also takes into account ongoing effectual operations with timestamp equal to or less than $ts$. For that, it gathers them from \textit{CurrUpdates} into a local copy \textit{currUpdates} (\Cref{update: aggValues: gather currUpdates}). This local copy is maintained throughout the traversal: done effectual operations (indicated by a $done$ flag set to true) are eliminated from it, as well as effectual operations with keys which are out of the range covered by the current subtree. The update is done using versioned writes (\Cref{update: aggValues: traverse: update}). To traverse its root-to-leaf path, it chooses each time whether to continue to the right or left child by simply searching for $k$, and then uses a versioned read with timestamp $ts$ to obtain the appropriate child (\Cref{update: aggValues: traverse: choose next}). 

The next step is to apply the operation to the target leaf area (\Cref{update: apply}): In case of an insertion, the child pointer from the parent to the target leaf (namely, the leaf at which the traversal arrived) is modified using a versioned write with $ts$ to point at a new internal node which has the new node and the target leaf as its children. In case of a deletion, we add a step that was not part of the base algorithm---to apply the deletion, the target leaf is marked as deleted by setting a $marked$ flag which we place in the tree leaves.
A deletion is finalized by pointing the target leaf's grandparent at the leaf's sibling using a versioned write with $ts$ only later---in \Cref{update: finalize del} (which is for deletions only).
After the operation is applied as detailed above, it sets its \textit{Update} object's $done$ field to true and removes the object from \textit{CurrUpdates} (\Cref{update: unannounce}). Lastly, it releases the lock as in the base algorithm---of the target leaf's original parent in case of an insertion or the target leaf's original grandparent in case of a deletion (\Cref{update: unlock}).
The reason for the special deletion scheme (using a mark in deletion, and ordering its steps to be marking the leaf, then setting \textit{done} in the \textit{Update} object and removing it from \textit{CurrUpdates}, then modifying the grandparent's pointer) is as follows:
Without the marking step, a deletion announcement would have been removed from \textit{CurrUpdates} while its leaf is still physically in the tree, and so \contains{} and a failing \ins{} might consider the node as in the tree. Instead, they check the \textit{marked} field in their \codestyle{isDeleted} call to correctly determine if the node is deleted, and thanks to the \textit{marked} field they do not miss a linearized deletion even if it has not yet physically unlinked the node from the tree. On the other hand we cannot physically unlink the node prior to 
unannouncing the deletion (setting its announcement's \textit{done} field and removing the announcement from \textit{CurrUpdates}), since if we did that then an operation $op$ that calculates aggregate metadata using \textit{CurrUpdates} (either an aggregate query, or an effectual operation that helps update aggregate metadata on behalf of other effectual operations) might gather from \textit{CurrUpdates} an announcement of a deletion whose parent is already unlinked and mistakenly consider it in later subtrees, as $op$ does not eliminate the deletion from its \textit{currUpdates} because it does not narrow the range according to the already-unlinked parent (which is solved by using the \textit{done} field to correctly eliminate the deletion from \textit{currUpdates}).

\subsubsection{Aggregate queries}\label{subsection: aggregate queries}
An aggregate query starts by obtaining a timestamp $ts$, and then continues like the sequential counterpart, running the traversals and \codestyle{computeAnswer} calls defined for the query, but with the following modifications to the template in \Cref{fig: basicAggregateQuery template} used to run the traversals (for which the template will take $ts$ as an input): 

\begin{enumerate}
\item\label{template modification: gather} As part of \textit{init}, gather \textit{CurrUpdates} with timestamp $\leq ts$ into a local copy \textit{currUpdates}, and if it is the first traversal ran for this query---also guarantee (in a way that will be detailed in \Cref{section: fastup,section: fastquer}) that no effectual operations other than the gathered ones will later obtain a timestamp $\leq ts$.
\item\label{template modification: eliminate out-of-range} After running \codestyle{shouldDescendRight}, eliminate out-of-range effectual operations (namely, effectual operations with key < \textit{node.key} if proceeding to the right, or with key $\geq$ \textit{node.key} if proceeding to the left) from \textit{currUpdates}.
\item\label{template modification: obtain versions} To obtain \textit{node.left}, \textit{node.right} and \textit{aggValue} of the obtained left child, access these fields through versioned reads with $ts$ and plug in the effect of the relevant effectual operations from \textit{currUpdates} (while traversing \textit{currUpdates}, ignore operations with \textit{done==true} and eliminate them from \textit{currUpdates}). The way their effect is plugged in will be detailed in \Cref{section: fastup,section: fastquer}.
\end{enumerate}

\subsubsection{\contains{} operation}

A \contains{}$(k)$ operation searches for $k$ until reaching a leaf $L$ as in the base tree algorithm, while additionally maintaining a \textit{prev} variable holding the previous traversed node and setting a \textit{direction} variable to left or right according to which child of \textit{prev} $L$ is. While searching, all reads of versioned \textit{left} and \textit{right} fields are standard reads (i.e., return the value in the first version in the list). If $L.key == k$, it calls \codestyle{isDeleted}(\textit{L, prev}) (see its description in \Cref{subsection: auxiliary methods}) and if \codestyle{isDeleted} returns false it returns $L$'s value else it returns NOT\_FOUND; otherwise ($L.key \neq k$)---it returns \codestyle{getValueIfInserted}(\textit{L, prev, direction, k}) (see its description in \Cref{subsection: auxiliary methods}).

\subsubsection{\codestyle{isDeleted} and \codestyle{getValueIfInserted} auxiliary methods}\label{subsection: auxiliary methods}

The \codestyle{isDeleted} and \codestyle{getValueIfInserted} auxiliary methods are used to help guarantee the correctness of non-effectual operations, by verifying whether the key they found to be inserted or deleted is not already considered deleted or inserted respectively. These methods' objective is to give an answer that is true at some point during the calling operation's interval, for a correct linearization. All their reads of versioned fields are standard reads (i.e., return the value in the first version in the list).

The method \codestyle{isDeleted} is called by \contains{} and \ins{} after they find the searched-for key in the tree, before they return its associated value or failure respectively, to verify that this key was not considered as already deleted by a deletion that has already taken a timestamp (thus considered linearized) but has not yet modified the tree structure. 
The method \codestyle{getValueIfInserted} is called by \contains{} and \del{} after they do not find the searched-for key, before they return a negative answer, to verify that this key was not considered as already inserted by an insertion that has already taken a timestamp (thus considered linearized) but has not yet linked the new node to the tree.
These methods' implementation presentation is deferred to \Cref{appendix: auxiliary methods}.

\section{Analysis}\label{section: analysis}
\subsection{Correctness}
Our algorithms are linearizable. Effectual operations and aggregate queries are linearized by timestamp order, where aggregate queries are linearized after effectual operations with the same timestamp, and effectual operations with the same timestamp in \fastup{} are ordered according to when the timestamp that is set in their announcement is obtained. The linearization of original tree operations that do not modify the tree (\contains{} and failing \ins{} and \del{}) is more delicate, and is detailed in \Cref{appendix: correctness}, with a linearizability proof.

\subsection{Time complexity}\label{subsection: complexity}

For each of the two presented algorithms, we analyze the time complexity of the new aggregate query operation, as well as the addition to the time complexity of the tree operations, incurred by our extension of the base algorithm.
In our analysis we denote the number of threads in the system by $t$, %
the number of effectual operations / aggregate queries whose execution interval overlaps with that of the analyzed operation by \textit{concUpdates} / \textit{concQueries}, and the number of threads running effectual operations whose execution interval overlaps with that of the analyzed operation by \textit{concUpdatingThreads} (this, unlike \textit{concUpdates}, is bounded by $t$). For an effectual operation, we denote by \textit{effectualDepth} the number of nodes traversed in its last traversal in \Cref{update: start base alg} in \Cref{subsection: ins and del} (the base algorithm might carry out multiple traversals due to restarts). For an aggregate query, we denote by \textit{queryDepth} the number of nodes traversed in its longest traversal (which is equal to the depth in the query's timestamp of the leaf reached in this traversal).

As shown in \Cref{appendix: complexity,section: optimizations}, when embedding our algorithms with the optimizations mentioned there, \fastup{} incurs no additional asymptotic time cost on any of the base tree's operations, and performs aggregate queries in $O(\textit{queryDepth}\cdot t\cdot\min\{\textit{concUpdates},\textit{concQueries}\})$ time. 
\fastquer{} performs aggregate queries in $O(\textit{queryDepth}\cdot\min\{\textit{concUpdates},\textit{concQueries}\})$ time, and incurs $O(\textit{concUpdatingThreads})$ additional time on \contains{} and failing \ins{} and \del{}, and $O(\textit{effectualDepth}\cdot\textit{concUpdatingThreads})$ on effectual operations. 

\section{Related work}\label{section: related work}

Aggregation has been studied in the database community, where aggregate queries like COUNT, SUM and AVG are used to calculate an aggregate value over several database rows. Our term definitions in \Cref{section: aggregate} are in the spirit of definitions from works on databases and consolidate them into a unified view on aggregate queries useful for addressing aggregate queries on concurrent trees.
\Cref{definition: additive aggregate function} resembles the definition of distributive aggregate function in \cite{GrayCBLRVPP97}, \Cref{alternative definition: additive aggregate function} resembles the definition of commutative-semigroup aggregation function in \cite{cohen2006rewriting}, and \Cref{definition: subtractive aggregate function} resembles the definition of distributive additive aggregate function in \cite{jurgens2002index}.

Various works have researched concurrent range queries that scan the keys in the range \cite{petrank2013lock,nelson2022bundling,arbel2018harnessing,brown2012range,avni2013leaplist,basin2017kiwi,winblad2018lock,chatterjee2017lock,fatourou2019persistent,wei2021constant,sheffi2022eemarq}, an approach which is too costly for implementing efficient concurrent aggregate queries that return more succinct information about a key range. A certain kind of concurrent aggregate queries has been addressed in \cite{sela2021concurrentSize}, which implements specifically a size query on sets and dictionaries. 
We study general concurrent aggregate queries, similarly tothe independent work of~\cite{kokorin2023wait}. Interestingly, \cite{kokorin2023wait} does not support a failure option for the \ins{} and \del{} operations. Such failures may require traversing the tree twice, which poses additional challenges.
In contrast, this work offers an improved space complexity over the current work, as they do not employ multi-versioning. Another advantage of \cite{kokorin2023wait} is that they also offer support for balanced trees.

Multi-versioning, where multiple versions of data structure objects are preserved when it is modified to enable queries to access old versions, has been employed in previous works which use multi-versioning to offer support for range queries (e.g., \cite{fatourou2019persistent,wei2021constant,sheffi2022eemarq,kobus2022jiffy}).
Our multi-versioning based design for concurrent trees with aggregate queries may be used to integrate also support for non-aggregate range queries, similarly to those works.
Using multi-versioning allows to reduce contention and interference between effectual operations and operations that do not modify the data structure: effectual operations leave versions for queries and do not need to help them further, and queries read old versions without interfering with newer updates. 
For example, the \contains{} operation in our design does not need to coordinate with concurrent operations or help them throughout its traversal (it might only need to look at the announced effectual operations once when it completes its traversal). %
Another advantage of using multi-versioning in our design is that it enables to support the more general form of aggregate queries, composing several simple aggregate queries when several serial traversals are required (e.g., for querying the median key of a given input key range). Thanks to using multi-versioning, we may support such composite queries by executing multiple simple aggregate queries over the tree and obtaining object versions baring the same timestamp in all of them.
Using multi-versioning requires employing appropriate garbage collection to reclaim unnecessary versions, which is an orthogonal problem that could be addressed using techniques from e.g. \cite{ben2021space,sheffi2022eemarq,wei2023practically}.

\section{Discussion}\label{section: discussion}
We addressed the problem of designing a concurrent efficient implementation for aggregate queries on trees. We formalized aggregate queries that could be efficiently supported on concurrent trees, presented a design that augments a concurrent binary search tree with such aggregate queries, and suggested two algorithms that implement this design, demonstrating a trade-off between aggregate query time and tree update operations.
It would be interesting to further investigate if there is a better alternative in this time trade-off between effectual operations and aggregate queries, as well as in the time-space trade-off, or whether it is impossible to achieve better solutions such as better aggregate query time complexity while still maintaining optimal time complexity for the original tree operations like our \fastup{} algorithm.
Another compelling research direction is a generalization of our design to additional trees, including lock-free, balanced and internal trees.

\appendix
\section{Aggregate query examples}\label{section: aggregate examples}

\illus{For illustration, we will take as a running example a tree of donations where a node's key is a donation amount (as an integer) and a node's value is a list of donors that gave a donation of this amount. 
For our metadata, we take $B$ to be $\mathbb{Z}$, and define $f(X)\coloneq \sum_{(key,value)\in X}key \cdot value.size$ (where $key$ is a donation amount and $value$ is a list of donors), namely, $f$ returns the sum of subtree donations.
$\oplus_f$ and $\ominus_f$ are simply addition and subtraction on $\mathbb{Z}$, and \codestyle{identityElement$_{(B,\oplus_f)}=0$}.}

The simplest aggregate queries are ones concerned with a key range of the form $(-\infty, k)$, such as \select{} and \rank{}, which may be answered using a simple aggregate query that requires a single traversal.
\illus{To illustrate a simple aggregate query that requires a single traversal on our running example, we inquire about the accumulated donation of all small donations up to a certain amount. For that we shall define a \codestyle{sumUpTo} query whose input would be the upper donation amount bound. 
The \codestyle{sumUpTo} query uses a single traversal to obtain the accumulated donation amount of all donations up to the given bound.
Its \codestyle{shouldDescendRight} returns true iff $input \geq key$. In this case the value of \codestyle{aggValueUpToCurrentKey} is not used, but there are queries which need it to decide on the traversal's direction. For instance, \select{} proceeds to the right child if $input$ (which is the required leaf's index) $>$ \codestyle{aggValueUpToCurrentKey} (which counts the leaves with key $\leq key$, a count that includes the leaves of the current node's left subtree).
The \codestyle{computeAnswer} method for \codestyle{sumUpTo} simply returns the aggregate value received from the traversal.}

Thanks to $f$ being subtractive, aggregate queries concerned with a key range of the form (\textit{lowKey,upKey}), such as \codestyle{sum}(\textit{lowKey,upKey}), may be computed using a simple aggregate query with two independent traversals, and then having \codestyle{computeAnswer} compute \textit{aggValueForUpKey}$\ominus_f$\textit{aggValueForLowKey} to get the aggregate value on the desired key range.
This could be extended to unions of such key ranges by having \codestyle{computeAnswer} apply $\oplus_f$ over the results of the sub-ranges. 
We note that \codestyle{computeAnswer} has an additional role other than combining the traversals' aggregate values into an aggregate value on the desired range: certain queries, such as average and sample variance queries, require additional calculations which \codestyle{computeAnswer} takes care of.
\illus{For example, to compute the average donation amount for a certain key range, we should hold different metadata: a $\langle sum,size\rangle$ tuple, and change $f$ to compute both sum and size. Traversals would return a tuple of $\langle sum,size\rangle$ values as their outputted aggregate value. \codestyle{computeAnswer} would compute the aggregate sum and size of the input key range by combining the traversals' outputted aggregate values using the $\oplus_f$ and $\ominus_f$ operators, and then divide the aggregate sum by the aggregate size to get the average.}

The most general aggregate queries we define are composite queries, that require serial execution of a chain of two or more queries one after another, as each query in the chain needs an input value which is the result of a previous query in the chain.
\illus{For example, we may run a \codestyle{medianKeyInRange}(\textit{lowKey,upKey}) query that would return the median key (or any other percentile as a generalization) of a given input key range. To implement such a query we may first run a simple aggregate query made of two independent \rank{} traversals: \rank{}(\textit{upKey}) and \rank{}(\textit{lowKey}). Let \textit{rankUp} and \textit{rankLow} be the aggregate values they return. Define this first simple aggregate query's \codestyle{computeAnswer} to return \textit{rankLow+(rankUp-rankLow)/2} (which is the index of the required median key). Then the second simple aggregate query should run a \select{} query on the output of the first one.}

\section{Implementation details}\label{section: implementation details}
\subsection{Node classes and the \codestyle{Update} class}\label{subsection: node and Update fields}
The fields of our node classes appear in \Cref{fig: node fields}. An internal node does not contain a value as we handle an external tree (in which the items are placed in the leaves). A leaf class does not include an \textit{aggValue} field since its aggregate value is constant and may be directly computed from its key and value as $f(\{(key,value)\})$.

\begin{figure}
\begin{lstlisting}
@\underline{class InternalNode}@:
    key
    aggValue
    left
    right
    leftRightLock
@\underline{class Leaf}@:
    key
    value
    marked
\end{lstlisting}
\caption{Fields of the node classes}\label{fig: node fields}
\end{figure}

Our \codestyle{Update} class fields appear in \Cref{fig: Update fields}.
\textit{edgeSource} contains the node whose \textit{edgeDirection} (left or right) child should be modified, namely, the grandparent node of the leaf to be deleted in case of a deletion, or the parent (which will become the grandparent after the insertion is complete) in case of an insertion.
\textit{edgeTarget} has the new child: the deleted node's sibling for a deletion, or the new internal node in case of an insertion.
\textit{leaf} is either the deleted node or the inserted new node.
\textit{operationKind} is either \del{} or \ins{}.

\begin{figure}
\begin{lstlisting}
@\underline{class Update}@:
    timestamp
    leaf
    edgeSource
    edgeTarget
    edgeDirection
    done
    operationKind
\end{lstlisting}
\caption{Fields of the \codestyle{Update} class}\label{fig: Update fields}
\end{figure}

\subsection{Versioned fields}\label{subsection: versioned field}
Both child pointers and aggregate value fields in the tree nodes in our algorithms are implemented using the \codestyle{VersionedField} class, whose pseudocode appears in \Cref{fig: VersionedField}. For child pointers the template parameter $T$, which indicates the value type of the field, is a pointer to a tree node. For aggregate values, $T$ is the aggregate type $B$.

\begin{figure}
\begin{lstlisting}
@\underline{class Version$\langle$T$\rangle$}@:
    T value
    int timestamp
    Version$\langle$T$\rangle$ next
@\underline{class VersionedField$\langle$T$\rangle$}@:
    Version$\langle$T$\rangle$ versionListHead;
    @\underline{read()}@:
        return versionListHead.value
    @\underline{standardTimestampedRead()}@:
        version = versionListHead
        return (version.value, version.timestamp)
    @\underline{versionedRead(ts)}@:
        version = versionListHead
        while version.timestamp > ts:
            version = version.next
        return version.value
    @\underline{write(v, ts)}@:
        versionListHead = new Version$\langle$T$\rangle$(v, ts, versionListHead)
    @\underline{writeIfTimestamp(lastTs, newValue, newTs)}@:
        firstVersion = versionListHead
        if firstVersion.timestamp == lastTs:
            versionListHead.CAS(firstVersion, new Version$\langle$T$\rangle$(newValue, newTs, firstVersion))
    \end{lstlisting}
\caption{The \codestyle{VersionedField} class}\label{fig: VersionedField}
\end{figure}

\subsection{\codestyle{isDeleted} and \codestyle{getValueIfInserted} auxiliary methods}\label{appendix: auxiliary methods}

The usage of the \codestyle{isDeleted} and \codestyle{getValueIfInserted} auxiliary methods is described in \Cref{subsection: auxiliary methods}. Next we bring their implementation. 

\begin{figure}[b]
\begin{lstlisting}
@\underline{isDeleted(\textit{leaf, parent})}@:
if @\textit{parent}@ is not permanently locked: @\label{isDeleted: not locked start}@
    return false @\label{isDeleted: not locked end}@
if a deletion @\textit{del}@ with the key @\textit{leaf.key}@ is found in @\textit{CurrUpdates}@: @\label{isDeleted: found delete start}@
    Guarantee @\textit{del.timestamp}@ is set@\label{isDeleted: guarantee ts}@
    return true @\label{isDeleted: found delete end}@
return @\textit{leaf.marked}@ @\label{isDeleted: marked}@
\end{lstlisting}
\caption{Pseudocode for the \codestyle{isDeleted} method}\label{fig: isDeleted}
\end{figure}

The method \codestyle{isDeleted} (\Cref{fig: isDeleted}) checks if the node from which the caller reached \textit{leaf} is not permanently locked (which should be the case for an internal node that is under deletion), and if that is the case, it returns false to indicate the key was not deleted at the time the calling operation obtained \textit{leaf} (\Crefrange{isDeleted: not locked start}{isDeleted: not locked end}).
Otherwise, an announcement of a deletion with the searched-for key is looked for in \textit{CurrUpdates}. If such an announcement is found, the caller makes sure that it has a set timestamp (as detailed in \Cref{subsection: the 2 algs}) and then returns true (\Crefrange{isDeleted: found delete start}{isDeleted: found delete end}).
Else, the method returns a value according to the \textit{marked} flag of \textit{leaf} (\Cref{isDeleted: marked}).

\begin{figure}
\begin{lstlisting}
@\underline{getValueIfInserted(\textit{leaf, parent, leafDirection, searchedKey})}@:
if the edge from @\textit{parent}@ in direction @\textit{leafDirection}@ is locked and an insertion @\textit{ins}@ with the key @\textit{searchedKey}@ is found in @\textit{CurrUpdates}@: @\label{getValueIfInserted: found insert start}@
    Guarantee @\textit{ins.timestamp}@ is set@\label{getValueIfInserted: guarantee ts}@
    return @\textit{ins.leaf.value}@@\label{getValueIfInserted: found insert end}@
@\textit{curr}@ = the child pointer from @\textit{parent}@ in direction @\textit{leafDirection}@@\label{getValueIfInserted: search start}@
while @\textit{curr}@ is not a leaf:
    if @\textit{searchedKey}@ < @\textit{curr.key}@:
        @\textit{curr = curr.left}@
    else:
        @\textit{curr = curr.right}@
if @\textit{curr.key}@ == @\textit{searchedKey}@:
    return @\textit{curr.value}@@\label{getValueIfInserted: return curr.value}@
return NOT_FOUND@\label{getValueIfInserted: search end}@
\end{lstlisting}
\caption{Pseudocode for the \codestyle{getValueIfInserted} method}\label{fig: getValueIfInserted}
\end{figure}

The method \codestyle{getValueIfInserted} (\Cref{fig: getValueIfInserted}) checks if the edge to the found leaf is locked (separately from the other edge from \textit{parent}, namely, this is not a permanent lock of \textit{parent}), and if so, \textit{CurrUpdates} is searched for an insertion announcement with \textit{searchedKey}. If such an announcement is found, the caller makes sure that it has a set timestamp (as detailed in \Cref{subsection: the 2 algs}) and then returns its value (\Crefrange{getValueIfInserted: found insert start}{getValueIfInserted: found insert end}).
In any other case (the edge from which the caller reached \textit{leaf} is not locked or there is no insertion announcement with \textit{searchedKey} in \textit{CurrUpdates}), a traversal (similar to the search performed by \contains{} in the base tree algorithm) is resumed from the current child of \textit{parent} in direction \textit{leafDirection} (\Crefrange{getValueIfInserted: search start}{getValueIfInserted: search end}).

\section{\fastup{} details}\label{section: fastup}

We fill in the missing details in the operations' implementation that were not covered so far. Throughout the description, we denote the id of an executing thread by \textit{tid}.
The tree holds several global members: the root node, an integer \textit{Timestamp} field which holds the current timestamp, and the \textit{CurrUpdates} array with a per-thread cell to be pointed at its ongoing effectual operation.
See \Cref{fig:fastup} for a visualization of the tree parts.

\begin{figure*}%
  \centering
  \includegraphics[width=1.1\textwidth]{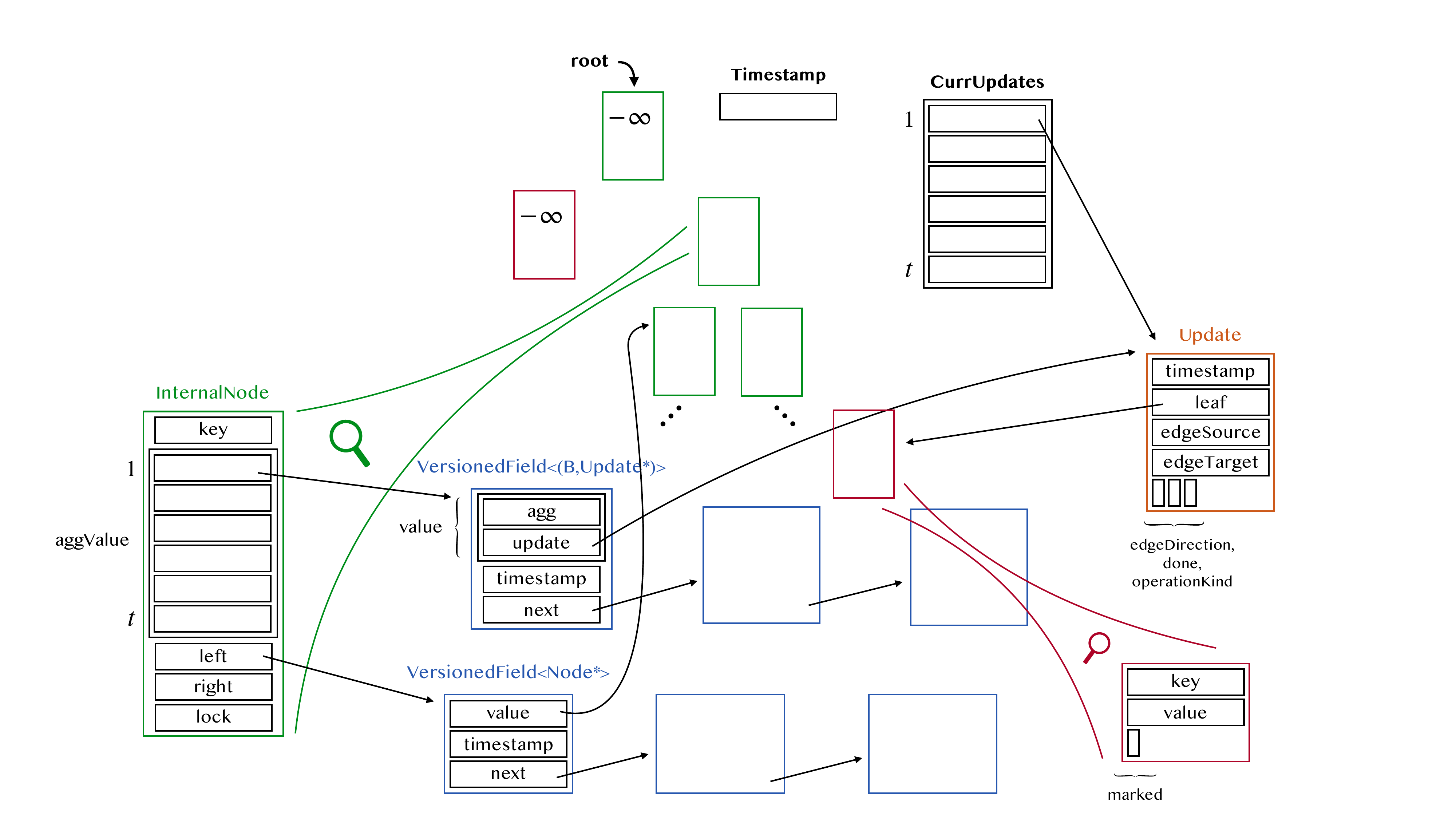}
  \caption{\fastup{} representation}
  \label{fig:fastup}
\end{figure*}

We start with describing the announcement mechanism and the way operations get their timestamp. An effectual operation $op$ announces itself by creating an \textit{Update} object (\Cref{fig: Update fields}) with its details (and a special \textit{NOT\_SET} value at the timestamp field) and pointing \textit{CurrUpdates[tid]} at it. It then performs \textit{CurrUpdates[tid].timestamp.CAS(NOT\_SET, Timestamp)} to set the timestamp of $op$ to the current global timestamp (if it has not yet been set by a concurrent aggregate query, hence the usage of \cas{} which stands for compare-and-swap). This completes \Cref{update: announce} in \Cref{subsection: ins and del}. 
To unannounce itself (\Cref{update: unannounce} in \Cref{subsection: ins and del}) it simply sets \textit{CurrUpdates[tid]} to NULL.

The operation could not have announced the \textit{Update} object with the timestamp field already set, because had it done so, it might have announced the operation with a too old timestamp, preceding the timestamps of aggregate queries that have already executed without considering $op$. This mechanism---of first announcing, then obtaining the global timestamp value, then setting it in the announcement's timestamp field---ensures that effectual operations will not be announced with old timestamps.
This together with the help by aggregate queries to set the timestamp if they observe it before it is set (as detailed next), enables aggregate queries to prevent effectual operations that they have not considered as preceding them from taking a timestamp smaller than theirs.

An aggregate query operation first of all grabs a timestamp $ts$ (see \Cref{subsection: aggregate queries}), which it accomplishes by executing a fetch-and-increment on the \textit{Timestamp} field (the fetch-and-increment atomically reads the old value---which is returned and will be the query's timestamp---and increments it by 1). 
In each of the operation's traversals, the operation goes through the array \textit{CurrUpdates} and gathers the objects pointed from it that have timestamp $\leq ts$ into \textit{currUpdates}, as mentioned in \Cref{template modification: gather} in the traversals' template modifications detailed in \Cref{subsection: aggregate queries}. In all the traversals this query runs but the first one, \textit{Update} objects with a \textit{NOT\_SET} timestamp value are considered as future ones and are disregarded. But the first traversal is responsible to guarantee that no effectual operations other than the ones it gathers will obtain a timestamp $\leq ts$. To this end, while it goes through \textit{CurrUpdates} to gather preceding effectual operations, for any object \textit{CurrUpdates[i]} it encounters whose timestamp value is \textit{NOT\_SET}, it carries out the following: It first executes \textit{CurrUpdates[i].timestamp.CAS(NOT\_SET, Timestamp)}. After this \cas{}, whether it was the \cas{} of this operation that succeeded, or it failed because a concurrent operation succeeded, \textit{CurrUpdates[i].timestamp} is guaranteed to be set. Therefore it reads its value again, and if it is $\leq ts$ it gathers it into \textit{currUpdates}.

\contains{} as well as failing \ins{} and \del{} operations (namely, ones that return failure) might also need to verify that a timestamp field of an \textit{Update} object they hold is set (in \Cref{isDeleted: guarantee ts,getValueIfInserted: guarantee ts} in the \codestyle{isDeleted} and \codestyle{getValueIfInserted} methods respectively). They do it similarly, by applying \textit{CAS(NOT\_SET, Timestamp)} to the \textit{timesatmp} field of the \textit{Update} object if its value is \textit{NOT\_SET}.

Effectual operations announce themselves in \fastup{} without serializing on some global variable and without being aware of which effectual operations precede them; each of them performs work only on behalf of itself (unlike in \fastquer{} where, as will be shown in \Cref{section: fastquer}, they need to help preceding ones). Therefore, \Cref{update: aggValues: gather currUpdates,update: aggValues: traverse: eliminate from currUps} in \Cref{subsection: ins and del} are void in \fastup{}, and in \Cref{update: aggValues: traverse: update} there is no usage of \textit{currUpdates}.

Moving on to the aggregate metadata representation, in \fastup{} an \textit{aggValue} field in a tree node is an array with a versioned field cell per thread. The value in this versioned field is made of two variables: \textit{agg}---the aggregate value, and \textit{update}---each version value will also hold a pointer to the \textit{Update} object announcing the operation that created this version. 
An effectual operation, with \textit{update} being a pointer to its \textit{Update} object, updates the aggregate value for each node in its traversal (\Cref{update: aggValues: traverse: update} in \Cref{subsection: ins and del}) as follows. It obtains the current aggregate value for this thread (which represents the value of $f$ on all elements currently in this node's subtree that were inserted by the current thread, before considering the current operation): \textit{currValue = node.aggValue[tid].read().agg}, calculates the new aggregate value:
\textit{newValue = currValue $\oplus_f$ f(\{(update.leaf.key,update.leaf.value)\})} if this is an insertion, or the same but with $\ominus_f$ instead of $\oplus_f$ if this is a deletion (in accordance with \Cref{property: update for additive,property: update for subtractive}); and finally performs a versioned write: \textit{node.aggValue[tid].write((newValue, update), update.timestamp)}.

An aggregate query needs to calculate the aggregate metadata value of left children along its traversal (see \Cref{template modification: obtain versions} in \Cref{subsection: aggregate queries}), which it does using its current \textit{currUpdates} as follows. Let \textit{left} be  the obtained left child of a node whose key is $k$.
The query first obtains for each thread its aggregate value in \textit{left} in the query's timestamp $ts$ (namely, the value of $f$ over elements in \textit{left}'s subtree inserted by this thread in timestamps $\leq ts$), and then computes the total aggregate value (over all elements in the subtree) from all the per-thread values. In detail, for each thread id $tid$, it performs a versioned read \textit{value$_{tid}$ = left.aggValue[tid].read(ts)}, and takes \textit{aggValue$_{tid}$ = value$_{tid}$.agg}, unless \textit{currUpdates[tid] $\neq$ NULL} and \textit{currUpdates[tid].leaf.key < k} and \textit{value$_{tid}$.update $\neq$ currUpdates[tid]} (which means an ongoing effectual operation by thread $tid$ on a leaf in \textit{left}'s subtree has not yet updated its aggregate value in \textit{left} so its effect should be computed using its announcement) in which case it takes \textit{aggValue$_{tid}$ = value$_{tid}$.agg $\oplus_f$ f(\{(currUpdates[tid].leaf.key, currUpdates[tid].leaf.value)\})}, or the same but with $\ominus_f$ instead of $\oplus_f$ if \textit{currUpdates[tid].operationKind} is \del{}.
It then computes the total aggregate value by $\oplus_f$\textit{aggValue$_{tid}$} over all \textit{tid}s.

To complete the algorithm's details, we explain how \textit{node.left} and \textit{node.right} are obtained during an aggregate query's traversal (see \Cref{template modification: obtain versions} in \Cref{subsection: aggregate queries}). 
To this end, an aggregate query with timestamp $ts$ performs the following (we describe how to obtain \textit{node.left}, the description for \textit{right} is similar): It checks if there exists an announcement \textit{update} in \textit{currUpdates} with \textit{update.sourceEdge == node} and \textit{edgeDirection == left}. 
(As an optimization, searching \textit{currUpdates} could be done only if the appropriate lock is not acquired.)
If so, it takes \textit{update.edgeTarget} as the desired value. Otherwise, it takes the result of a versioned read of \textit{node.left} with $ts$.

\section{\fastquer{} details}\label{section: fastquer} 

As for \fastquer{}, its global members are the root node and the \textit{CurrUpdates} queue.
Effectual operations enqueue themselves to the queue to announce themselves (\Cref{update: announce} in \Cref{subsection: ins and del}), and the queue's enqueue operation takes care of setting the operation's timestamp to be the timestamp of the last operation in the queue +1. The wait-free queue of \cite{kogan2011wait}, which is based on \cite{michael1996simple}, may be used as a basis for the \textit{CurrUpdates} queue. It could be easily extended to contain consecutive timestamps in the nodes (by setting the timestamp field to the value in the current tail node +1 and trying to enqueue this node using a \cas{}). It could also be extended to enable the removal of a node that is not in the head of the queue, as should be done to unannounce an effectual operation (\Cref{update: unannounce} in \Cref{subsection: ins and del}).

Aggregate queries obtain a timestamp by reading the timestamp in the \textit{Update} object in the current last node of the queue (which is the last one to be added to the queue).
To gather announcements with timestamp $\leq ts$ from \textit{CurrUpdates} into a local \textit{currUpdates} (by either an effectual operation with timestamp $ts$ in \Cref{update: aggValues: gather currUpdates} in \Cref{subsection: ins and del}, or an aggregate query with timestamp $ts$ as mentioned in \Cref{template modification: gather} in the traversals' template modifications detailed in \Cref{subsection: aggregate queries}), the operation goes through \textit{CurrUpdates}'s nodes and copies their announcements' pointers starting from the first (oldest) node until reaching a node with timestamp $>ts$, whose announcement is excluded from the gathered announcements.
In a similar fashion, when \contains{} and failing \ins{} and \del{} operations traverse \textit{CurrUpdates} (in \Cref{isDeleted: found delete start,getValueIfInserted: found insert start} in the \codestyle{isDeleted} and \codestyle{getValueIfInserted} methods respectively), they in fact first obtain the current global timestamp by reading the timestamp in the \textit{Update} object in the current last node of the queue, and then go through \textit{CurrUpdates}'s nodes and copy their announcements' pointers starting from the first node until reaching a node with timestamp greater than the timestamp they obtained (otherwise, they might continue going through \textit{CurrUpdates} forever as more and more concurrent operations enqueue their announcements).

Unlike in \fastup{}, in \fastquer{} an aggregate query with timestamp $ts$ needs to do nothing to guarantee that effectual operations will not later obtain a timestamp $\leq ts$, as after it obtains the timestamp $ts$, any later operation will be allocated a bigger timestamp. Hence, guaranteeing that in the first query's traversal as mentioned in \Cref{template modification: gather} in \Cref{subsection: aggregate queries} is void in \fastquer{}.
Also, for \contains{} and failing \ins{} and \del{} operations, guaranteeing a timestamp in an \textit{Update} object they obtained from \textit{CurrUpdates} is set (in \Cref{isDeleted: guarantee ts,getValueIfInserted: guarantee ts} in the \codestyle{isDeleted} and \codestyle{getValueIfInserted} methods respectively) is void in \fastquer{}, since nodes are enqueued to the \textit{CurrUpdates} queue with a set timestamp.

We move on from the timestamp and announcement handling in \fastquer{} to its aggregate value access scheme.
To update the aggregate value of a node to reflect an effectual operation with timestamp $ts$ (\Cref{update: aggValues: traverse: update} in \Cref{subsection: ins and del}), it needs to add a version with timestamp $ts$ and the appropriate aggregate value to the linked list of versions in \textit{node.aggValue}. However, it must first update the aggregate value on behalf of ongoing effectual operations with smaller timestamps. To update the aggregate metadata on behalf of those effectual operations that have not yet updated it and on behalf of itself, it performs the following until going through \textit{currUpdates} (which is ordered by timestamp like the \textit{CurrUpdates} queue, and includes the current operation itself in the last announcement) to its end: executing a standard timestamped read of \textit{node.aggValue} to obtain its \textit{(lastValue, lastTs)}; continuing to go through \textit{currUpdates} (from the last point reached at the previous iteration) until reaching an object \textit{update} with a timestamp greater than \textit{lastTs}; calculating the new aggregate value taking \textit{update} into account by computing \textit{newValue = lastValue $\oplus_f$ f(\{(update.leaf.key, update.leaf.value)\})} if \textit{update.operationKind} is \ins{}, or the same but with $\ominus_f$ instead of $\oplus_f$ if it is \del{}; and trying to perform a thread-safe write of \textit{newValue} with timestamp \textit{update.timestamp} given the current version list's timestamp is still \textit{lastTs} (namely, running \textit{node.aggValue.writeIfTimestamp(lastTs, newValue, update.timestamp)}, for which there is no need to check for failure as it fails only if another thread performed the same write).

An aggregate query with timestamp $ts$ needs to calculate the aggregate metadata value of left children along its traversal (see \Cref{template modification: obtain versions} in \Cref{subsection: aggregate queries}), which it does using its current \textit{currUpdates} as follows. Let \textit{left} be  the obtained left child of a node whose key is $k$.
The query first executes a standard timestamped read of \textit{left.aggValue} to obtain its current \textit{(value, \textit{currTs})}. If $\textit{currTs}==ts$ it returns \textit{value}; if $\textit{currTs}>ts$, it executes a versioned read of \textit{left.aggValue} with $ts$ and returns the the returned value. Otherwise, it needs to plug in the effect of relevant ongoing operations on top of \textit{value}. For that, it goes through \textit{currUpdates}, and for each object \textit{update}, if \textit{update.leaf.key < k} and \textit{update.timestamp > ts} then it calculates the new aggregate value taking \textit{update} into account by computing \textit{value = value $\oplus_f$ f(\{(update.leaf.key, update.leaf.value)\})} if \textit{update.operationKind} is \ins{}, or the same but with $\ominus_f$ instead of $\oplus_f$ if it is \del{}. After exhausting \textit{currUpdates}, \textit{value} has the required aggregate metadata value.
As for how an aggregate query obtains \textit{node.left} and \textit{node.right} during its traversals, it does so the same as in \fastup{} (see details in the end of \Cref{section: fastup}).

\section{Correctness}\label{appendix: correctness}
\fastup{} and \fastquer{} are linearizable.
In this section we detail the linearization points of their operations, and use them to prove linearizability.

\subsection{Linearization Points}\label{section: linearization points}
Effectual operations are ordered by their timestamps. For \fastup{} we set the linearization point of an effectual operation to be at the moment the timestamp that is eventually set in its \textit{Update} announcement using a \cas{} is obtained (by either the operation itself or a helping aggregate operation); for \fastquer{} we set it at the linearization point of the enqueue of its \textit{Update} object to the \textit{CurrUpdates} queue (which is, by definition of linearization points of the underlying queue \cite{kogan2011wait}, when the node with this \textit{Update} object is physically linked to the queue).
Aggregate queries are also ordered by timestamp, such that they occur after effectual operations of the same timestamp. We accordingly set their linearization points to be at the fetch-and-increment of the \textit{Timestamp} field in \fastup{}, and in \fastquer{} -- when they obtain in their beginning the pointer to the last node of the \textit{CurrUpdates} queue, from which they obtain the global timestamp.

Lastly, \contains{} and failing \ins{} and \del{} operations are linearized according to their last \codestyle{isDeleted} or \codestyle{getValueIfInserted} call (\contains{} always has only one such call) as follows. Let $op$ be such an operation on a key $k$, $m$ be the last method it called (\codestyle{isDeleted} or \codestyle{getValueIfInserted}), and $L$---the leaf node $op$ passes to $m$.
$op$ is linearized at a linearization point as follows:
\begin{enumerate}
    \item If $m$ is \codestyle{isDeleted}, and it returns false: a moment within $op$'s interval which is after the linearization point of the insertion of $L$ and before the linearization point of the deletion of $L$ (if such a deletion occurs in the execution).
    \item If $m$ is \codestyle{isDeleted}, and it returns true: a moment within $op$'s interval which is after the linearization point of the deletion of $L$ and before the linearization point of the next insertion of the key $k$ (if such an insertion occurs in the execution).
    \item If $m$ is \codestyle{getValueIfInserted}, and it returns the value that is in a node $N$ ($N$ is \textit{ins.leaf} if $m$ returns in \Cref{getValueIfInserted: found insert end}, or \textit{curr}'s value when \Cref{getValueIfInserted: return curr.value} is executed if $m$ returns in \Cref{getValueIfInserted: return curr.value}): a moment within $op$'s interval which is after the linearization point of the insertion of $N$ and before the linearization point of the deletion of $N$ (if such a deletion occurs in the execution).
    \item If $m$ is \codestyle{getValueIfInserted}, and it returns NOT\_FOUND: a moment within $op$'s interval in which $k$ is logically not found in the tree; namely, if an insertion of $k$ to the tree has been linearized at any point in the execution by the time $op$ returns, then a moment within $op$'s interval in which the last operation on $k$ to be linearized was a deletion, otherwise any moment within $op$'s interval is fine.    
\end{enumerate}
We should show the last linearization points are well-defined, as we shall do in \Cref{lemma: well defined lps}.

\subsection{Linearizability Proof}\label{section: linearizability proof}
We prove that our algorithms are linearizable using the equivalent definition of linearizability that is based on linearization points (see \cite[Section 7]{sela2021linearizabilityfull} and the atomicity definition in \cite{lynch1996distributed}). We need to show that (1) each linearization point occurs within the operation's execution time, and that (2) ordering an execution's operations (with their results) according to their linearization points forms a legal sequential history.
(1) is pretty straightforward for the linearization points we defined, although the linearization point definition for \contains{} and failing \ins{} and \del{} operations does require proving that they are well defined, as we do in \Cref{lemma: well defined lps}. We prove (2) in \Cref{claim: results comply}.

\begin{lemma}\label{lemma: well defined lps}
The linearization points of \contains{} and failing \ins{} and \del{} operations as defined in \Cref{section: linearization points} are well defined. Namely, a moment as described in them indeed occurs within $op$'s interval.
\end{lemma}
\begin{proof}
We use the same notations as in \Cref{section: linearization points}. Additionally, let $P$ be the parent node from which $op$ reached $L$ and which it passes to $m$.

We start with the case of $m$ being \codestyle{isDeleted} which returns false.
Let $t$ be the moment in which $op$ obtains $L$ from its parent. The \ins{} operation that inserts $L$ is already linearized at time $t$, because the fact that $op$ reaches $L$ in its traversal means that this \ins{} operation must have already physically linked $L$ to the tree, which it does only after its linearization point. It remains to show that the \del{} operation of $L$ (if such a deletion occurs in the execution) is not yet linearized at time $t$. We look at two cases according to the line at which $m$ returns. If it returns in \Cref{isDeleted: not locked end}, it observes $P$ when it is not permanently locked.
We observe that in the base tree algorithm, a tree node may be pointed only from its current parent in the tree and permanently-locked nodes that were its parent in the past (before they were unlinked from the tree when the node's then-sibling was deleted). When a node is deleted, the single node pointing to it that was not permanently locked becomes permanently locked as well, namely, all nodes pointing to a deleted node are permanently locked. Back to our proof, this implies that $L$ could not be deleted when the not-permanently-locked $P$ pointed to it.
Alternatively, $m$ returns in \Cref{isDeleted: marked}. This means that $m$ observed no announcement of a deletion of $k$ in \textit{CurrUpdates}, and then observed \textit{L.marked==false}. But from the moment the deletion of $L$ is linearized, either its announcement is found in \textit{CurrUpdates} or \textit{L.marked==true} (as the deletion's linearization point occurs when the announcement is in \textit{CurrUpdates}, and the announcement is removed from \textit{CurrUpdates} only after \textit{L.marked} is permanently set to true). So $L$'s removal was not linearized when $op$ obtained $L$.

Next, for the case of $m$ being \codestyle{isDeleted} which returns true: 
The linearization point of the deletion of $L$ obviously happens before the linearization point of the next insertion of the key $k$ (if such an insertion occurs in the execution). $op$ returns true only after the deletion of $L$ is linearized. In addition, when $op$ is invoked, $L$ is still reachable in the tree (as $op$ reaches it), so it is impossible that $L$ was removed and $k$ was re-inserted before $op$'s invocation.

In the case of $m$ being \codestyle{getValueIfInserted} which returns $N.value$: 
The linearization point of the insertion of $N$ obviously happens before the linearization point of the deletion of $N$ (if such a deletion occurs in the execution). $op$ returns only after the insertion of $N$ is linearized as it either guarantees its linearization in \Cref{getValueIfInserted: guarantee ts} or reaches $N$ through a traversal which means $N$ has been physically linked to the tree which happens only after its insertion's linearization. 
On the other direction, we show that $op$ is invoked before the deletion of $N$ is linearized: We observe that two effectual operations on the same key could not be announced in \textit{CurrUpdates} at the same time, due to the parent's lock they hold while being announced. Therefore, in case $op$ observes an announcement regarding the insertion of $N$ then at that moment $N$'s deletion could not have yet happened; and in case $op$ reached $N$ through a traversal, then $N$ has been physically linked to the tree beforehand, and certainly after the invocation of $op$ (otherwise $op$ would have reached $N$ during its traversal prior to calling \codestyle{getValueIfInserted}, which it did not---recall $m$ was called after not finding $k$), and at that moment $N$'s insertion was ongoing so a deletion of $N$ could not have been announced yet in \textit{CurrUpdates} so it was not yet linearized.

The last case revolves around $m$ being \codestyle{getValueIfInserted} which returns NOT\_FOUND.
Assuming that an insertion of $k$ has been linearized at any point in the execution by the time $op$ returns, we need to show that there is a moment during $op$'s interval in which the last operation on $k$ to be linearized was a deletion.
Assume, for sake of contradiction, that an insertion of a node $N$ with the key $k$ has been linearized prior to $op$'s invocation, and a deletion of $N$ is not linearized until $op$ returns.
Thus, from the moment $N$ is physically linked to the tree it cannot be physically unlinked during $op$'s interval. $N$ must not yet be linked when $op$ obtains $L$, otherwise $op$ would reach $N$ that has the searched-for key $k$ instead of reaching $L$ that has another key. This means that at the invocation moment of $op$, $P$ must be locked for inserting $k$. Thus, $N$ must be linked as a sibling of $L$, with a new internal node becoming $P$'s child instead of $L$. From the moment $N$ is linked, it remains in $P$'s subtree because the insertion and deletion operations of the base tree preserve ancestor-descendant relations for non-deleted nodes (so nodes may be added or removed between $P$ and $N$, but $N$ remains $P$'s descendant the whole time). And it must be physically linked prior to $m$'s traversal in \Crefrange{getValueIfInserted: search start}{getValueIfInserted: search end} because before the traversal, when \Cref{getValueIfInserted: found insert start} is executed, the \textit{Update} object announcing $N$'s insertion is already not found in \textit{CurrUpdates} anymore. Hence, $m$'s traversal from $P$ must reach $N$, which is a contradiction to returning NOT\_FOUND.
\end{proof}

\begin{claim}\label{claim: results comply}
Consider a sequential history formed by ordering an execution's operations (with their results) according to their linearization points defined in \Cref{section: linearization points}. 
Then operation results in this history comply with the sequential specification of a dictionary.
\end{claim}
\begin{proof}
The abstract state of the dictionary includes all keys on which the last operation was a successful \ins{}, with the value it inserted. We will show that all operations consider this as the state of the dictionary in their linearization point and return a suitable result.

Starting with \contains{} and failing \ins{} and \del{} operations, from their linearization point definition it immediately follows that their result complies with a dictionary's sequential specification:
\begin{enumerate}
    \item The linearization point of both \contains{}($k$) that returns a node's value after calling \codestyle{isDeleted} and a failing \ins{}($k$, $v$) is defined to be after the linearization point of the insertion of the found key and before its deletion, namely, in a moment $k$ is in the dictionary with the found value.
    \item The linearization point of \contains{}($k$) that returns NOT\_FOUND after calling \codestyle{isDeleted} is after the linearization point of the deletion of the node found with $k$ and before the linearization point of the next insertion of the key $k$, namely, in a moment $k$ is not found in the dictionary with the found value.
    \item The linearization point of \contains{}($k$) that returns a node's value after calling \codestyle{getValueIfInserted} is after the linearization point of the insertion of the newly-inserted node whose value is returned and before the linearization point of its deletion, namely, in a moment $k$ is in the dictionary with the found value.
    \item The linearization point of both \contains{}($k$) that returns NOT\_FOUND after calling \codestyle{getValueIfInserted} and a failing \del{}($k$) is defined to be when $k$ is not found in the dictionary.
\end{enumerate}

For a successful \del{}($k$), in its linearization point, a node $L$ with $k$ is physically found in the tree, and the deletion operation has permanently locked $L$'s parent. So there are no other operations on $k$ currently announced in \textit{CurrUpdates} (because they need to grab the lock first), which means any operation on $k$ that has already linearized has physically applied itself to the tree, and any operation on $k$ that is not yet linearized has not physically applied itself to the tree. Therefore, as $k$ is physically in the tree it means that the last effectual operation on $k$ was an insertion of $L$.

For a successful \ins{}($k$, $v$), in its linearization point, a node with $k$ is not physically found in the tree, and the insertion operation is holding a lock of the edge where the key should be inserted. So there are no other operations on $k$ currently announced in \textit{CurrUpdates} (because they need to grab the lock in the same location first), which means any operation on $k$ that has already linearized has physically applied itself to the tree, and any operation on $k$ that is not yet linearized has not physically applied itself to the tree. Therefore, as $k$ is physically not in the tree it means that the last effectual operation on $k$ was a deletion.
 
It remains to prove that an aggregate query $agg$ returns the correct result. For that we need to show that it obtains the logical values of the child pointers and the aggregate metadata at the moment of its linearization.
It obtains the value from the relevant tree node together with the effect of relevant effectual operations in \textit{currUpdates}.
Let $N$ be the current node whose child pointer or aggregate value $agg$ obtains during its traversal.
$agg$ plugs in only the effect of effectual operations which precede it and are not yet reflected in the current field value in $N$ thanks to the following observations regarding $agg$'s \textit{currUpdates} local variable at the current moment:
First, all effectual operations in \textit{currUpdates} are linearized before $agg$. That is because as detailed in \Cref{template modification: gather} in \Cref{subsection: aggregate queries}, $agg$ gathers ongoing announcements regarding effectual operations with $ts\leq agg.ts$, which are the ones that are linearized before it.
Second, the leaves of all effectual operations in \textit{currUpdates} are in $N$'s subtree. That is thanks to eliminating out-of-range effectual operations throughout the traversal as detailed in \Cref{template modification: eliminate out-of-range} in \Cref{subsection: aggregate queries}, and additionally thanks to ignoring effectual operations with \textit{done==true}---so that in case a deleted node's parent was unlinked and was not taken into account while narrowing the range of \textit{currUpdates}, the deleted node will not be mistakenly taken into account when calculating aggregate metadata of later nodes in the traversal.
Third, if an operation that is still in \textit{currUpdates} already updated a certain affected metadata or a child pointer, it will not be updated again on behalf of the same operation (thanks to a unique timestamp for effectual operations in \fastquer{}, and to the examination of the \textit{update} pointer field in \fastup{}).
We also observe that there are no ongoing effectual operations with $ts\leq agg.ts$ that $agg$ does not gather, thanks to $agg$ guaranteeing that no effectual operations other than the gathered ones will later obtain a timestamp $\leq agg.ts$ as mentioned in \Cref{template modification: gather} in \Cref{subsection: aggregate queries}.
Lastly, there is no effectual operation preceding $agg$ that is not in its \textit{currUpdates} and has not yet updated the tree (both the affected metadata values and the target area), because effectual operations leave \textit{CurrUpdates} only after updating all relevant fields in versions with their timestamp.

\end{proof}

\section{Time complexity}\label{appendix: complexity}
Using the notations brought in \Cref{subsection: complexity}, we next analyze the time complexity of each of our two algorithms.

\subsection{\fastup{} complexity}\label{subsection: fastup complexity}

In \fastup{}, searching \textit{CurrUpdates} in the \codestyle{isDeleted} and \codestyle{getValueIfInserted} methods costs $O(t)$ time. 
However, this cost could be eliminated
by setting aside some bits in each internal node's left-child lock as well as right-child lock (the locks that are placed together in the same memory word) for indicating the tid of the locking thread and modifying the locking mechanism of the base algorithm accordingly (to also set the tid by the same \cas{} that acquires a lock). Then \codestyle{isDeleted} and \codestyle{getValueIfInserted} could avoid searching \textit{CurrUpdates} and instead directly check the cell of the specified locking thread. 
Then we get zero additional asymptotic time cost on \contains{} and failing \ins{} and \del{}.

As for effectual operations, \fastup{} incurs $O(\textit{effectualDepth}\cdot\textit{concUpdates})$ additional time on them: They pay $O(1)$ time for announcing and unannouncing themselves (\Cref{update: announce,update: unannounce} in \Cref{subsection: ins and del}) as well as for applying the operation (\Cref{update: apply,update: finalize del}). To update affected aggregate values (\Cref{update: aggValues}), they go through $O(\textit{effectualDepth})$\footnote{An effectual operation goes through $O(\textit{effectualDepth})$ nodes during the traversal in \Cref{update: aggValues: traverse}, because after it reaches a leaf with key $k$ and locks the edge from its parent in case of \ins{} or permanently locks the parent in case of \del{} (in \Cref{update: start base alg}), searching for $k$ again from the root (in \Cref{update: aggValues: traverse}) may result in traversing through less nodes due to concurrent deletions, but not through additional nodes (which are impossible in the current operation's path because internal nodes are inserted only between a leaf and its parent, and the edge from this path's penultimate node to the leaf may not be modified during the operation's execution as it either the edge's lock is acquired or the penultimate node is permanently locked).} nodes, and for each, perform \Cref{update: aggValues: traverse: update} that costs $O(1)$ time and \Cref{update: aggValues: traverse: choose next} whose versioned read costs $O(\textit{concUpdates})$ time. (\Cref{update: aggValues: gather currUpdates,update: aggValues: traverse: eliminate from currUps} in \Cref{subsection: ins and del} are void in \fastup{} where effectual operations do not help other effectual operations and do not use \textit{currUpdates}.)
The \fastup{} algorithm could however be optimized to avoid any additional asymptotic time cost on effectual operations, by keeping record of the traversed nodes during \Cref{update: start base alg}, so that the recorded nodes will be traversed in \Cref{update: aggValues: traverse} and the versioned reads in \Cref{update: aggValues: traverse: choose next} will be avoided.

An aggregate query in \fastup{} takes $O(\textit{queryDepth}\cdot t\cdot\textit{concUpdates})$ time, according to the following calculation.
We assume all \codestyle{shouldDescendRight} and \codestyle{computeAnswer} methods defined by the query take $O(1)$ time, otherwise their cost should be counted as well.
Obtaining a timestamp for the query takes $O(1)$ time. Each traversal ran by the query goes through \textit{CurrUpdates} in its initialization (in \Cref{template modification: gather} in \Cref{subsection: aggregate queries}) in $O(t)$ time. The traversal then performs $O(\textit{queryDepth})$ iterations, where each iteration costs as follows, based on the traversals' template modifications detailed in \Cref{subsection: aggregate queries}:
Going through \textit{currUpdates} in \Cref{template modification: eliminate out-of-range} takes $O(t)$ time. As for \Cref{template modification: obtain versions}---obtaining child nodes takes $O(t)$ time to go through \textit{currUpdates} (to look for an ongoing update of the child pointer, which could be spared if indicating the locking thread's tid in internal nodes' locks as suggested above, but shaving the $t$ factor here anyhow does not improve the time cost of an aggregate query as this is not its dominant factor) and additional $O(\textit{concUpdates})$ time to perform a versioned read of the pointer, and computing the aggregate value of the obtained left child takes $O(t\cdot\textit{concUpdates})$ time because for each thread the query performs a versioned read on its cell in the metadata array (and also checks the thread's cell in \textit{concUpdates} in $O(1)$ time).

\subsection{\fastquer{} complexity}\label{subsection: fastquer complexity}

In \fastquer{}, searching the \textit{CurrUpdates} queue in the \codestyle{isDeleted} and \codestyle{getValueIfInserted} methods costs $O(\textit{concUpdatingThreads})$ time. In fact, each traversal of \textit{CurrUpdates} by any operation takes $O(\textit{concUpdatingThreads})$ time because each thread may have a single ongoing effectual operation at any moment, and as the queue is always traversed from the head up to a pre-obtained timestamp, only operations that were ongoing when the timestamp was obtained may be traversed. This means that \fastquer{} incurs $O(\textit{concUpdatingThreads})$ additional time on \contains{}. As for \ins{} and \del{}, they might call \codestyle{isDeleted} and \codestyle{getValueIfInserted} multiple times in \Cref{update: start base alg} in \Cref{subsection: ins and del}, as they may restart this step multiple times. However, it is easy to eliminate the additional cost from all calls except the first one by having the first call record the \textit{Update} object it found in \textit{CurrUpdates} (or record it did not find an object with the requested key),
and the following calls could avoid the search and just use what it recorded instead. This way, \fastquer{} will incur $O(\textit{concUpdatingThreads})$ additional time on failing \ins{} and \del{} as well.

For an effectual operation, other than the $O(\textit{concUpdatingThreads})$ additional time due to possibly calling \codestyle{isDeleted} and \codestyle{getValueIfInserted}, it pays the following additional time.
Announcing and unannouncing itself (\Cref{update: announce,update: unannounce} in \Cref{subsection: ins and del}) by enqueueing and removing respectively from the \textit{CurrUpdates} queue based on the wait-free queue of \cite{kogan2011wait} takes $O(t)$ time, which could be optimized to $O(\textit{concUpdatingThreads})$ using techniques of \cite{afek1995wait} as mentioned by \cite{kogan2011wait}.
Applying the operation (\Cref{update: apply,update: finalize del}) costs $O(1)$ time. To update affected aggregate values (\Cref{update: aggValues}), it first goes over \textit{CurrUpdates} to gather announcements into \textit{currUpdates} in $O(\textit{concUpdatingThreads})$ time (\Cref{update: aggValues: gather currUpdates}), then
goes through $O(\textit{effectualDepth})$ nodes (for the same argument that was made above for \fastup{}), and for each node: It performs \Cref{update: aggValues: traverse: update,update: aggValues: traverse: eliminate from currUps} in $O(\textit{concUpdatingThreads})$ time due to going over \textit{CurrUpdates}. In addition, it runs \Cref{update: aggValues: traverse: choose next} in $O(\textit{concUpdates})$ time due to the versioned read, but this step could be eliminated together with its cost, by keeping record of the traversed nodes during \Cref{update: start base alg}, so that the recorded nodes will be traversed in \Cref{update: aggValues: traverse}.

An aggregate query in \fastquer{} takes $O(\textit{queryDepth}\cdot\textit{concUpdates})$ time, according to the following calculation (assuming \codestyle{shouldDescendRight} and \codestyle{computeAnswer} methods defined by the query take $O(1)$ time).
Obtaining a timestamp for the query takes $O(1)$ time. 
Each traversal ran by the query goes through \textit{CurrUpdates} in its initialization (in \Cref{template modification: gather} in \Cref{subsection: aggregate queries}) in $O(\textit{concUpdatingThreads})$ time. The traversal then performs $O(\textit{queryDepth})$ iterations, where each iteration costs $O(\textit{concUpdates})$ because this is the cost of \Cref{template modification: eliminate out-of-range,template modification: obtain versions} in the traversals' template modifications detailed in \Cref{subsection: aggregate queries}, due to going through \textit{currUpdates} in $O(\textit{concUpdatingThreads})$ time which is bounded by $O(\textit{concUpdates})$, and performing versioned reads in $O(\textit{concUpdates})$ time.

In \Cref{section: optimizations} we suggest further optimizations that reduce the \textit{concUpdates} factor in the time complexity of aggregate queries down to $\min\{\textit{concUpdates},\textit{concQueries}\}$.

\section{Algorithm optimizations}\label{section: optimizations}

In addition to the optimizations mentioned in \Cref{appendix: complexity}, 
we suggest several more ways to optimize our algorithms.

Aggregate queries perform versioned reads of both child pointers and aggregate metadata, and each such read costs $O(\textit{concUpdates})$ time.
To reduce the \textit{concUpdates} factor in aggregate queries time complexity down to $\min\{\textit{concUpdates},\textit{concQueries}\}$, we may apply the following optimizations.

In \fastup{}, instead of writing to a versioned field by adding a new version to the head of its linked list of versions, we may update the current version in case of a write with the same timestamp, to spare the creation of redundant versions that no aggregate query is going to need.
The problem that might be then created is that aggregate queries that run concurrently with the effectual operation might not know whether the operation is already reflected in the current tree data or not and they should plug in its effect on their own.
For left and right child pointers, aggregate queries may simply solve the problem by taking into account an announced effectual operation's new edge target if it has an edge source they are currently examining, since that is for sure the correct value.
But for the aggregate metadata field, we need to make some changes to enable an aggregate query to correctly determine if an announced operation has already updated the examined metadata or not: We add an \textit{updateNum} array to the tree, with a cell per thread holding the last effectual operation id of the thread (a field read and written only by the thread, incremented by the thread on each effectual operation it runs); \textit{updateNum} and \textit{valueToSet} fields to the \textit{Update} object; and \textit{updateNum} and \textit{lastUpdateNum} fields to the \textit{VersionedField} object of the aggregate metadata (while having its value field contain only the aggregate value and not an update pointer as in the initial design, which solves also the problem of keep holding a pointer to the \textit{Update} object after it is no longer necessary, thus delaying its reclamation).
We have an effectual operation update the metadata as follows (where \textit{update} is its \textit{Update} object, \textit{version} is the \textit{VersionedField} object of the aggregate metadata it currently updates, and \textit{newValue} is the new value to be written to the version---reflecting the current value combined with the effect of the current effectual operation):%
\begin{lstlisting}
update.valueToSet = newValue
version.updateNum = update.updateNum
version.value = update.valueToSet 
version.lastUpdateNum = version.updateNum
\end{lstlisting}
We have aggregate queries run the following to obtain the correct metadata value (where \textit{update} is the \textit{Update} object obtained from its \textit{currUpdates} as it has the same key as the node whose aggregate value is calculated and \textit{version} is the \textit{VersionedField} object of the node for which the aggregate metadata is calculated):
\begin{lstlisting}
versionValue = version.value
if version.updateNum < update.updateNum:
    return < versionValue with the effect of update plugged in > 
updateValue = update.valToSet 
if version.lastUpdateNum < update.updateNum: return updateValue 
return version.value
\end{lstlisting}

In \fastquer{}, we suggest to apply two changes in order to be able to avoid creating versions not required by any ongoing aggregate query, thus reducing both the time complexity of aggregate queries and the tree's space consumption. 
First, aggregate queries will announce themselves (specifically, just their timestamp) in a global \textit{QueryTimestamps} wait-free queue, similar to the \textit{CurrUpdates} queue holding the effectual operations' announcements.
An effectual update will gather all announced query timestamps smaller than its own timestamp (similarly to the way ongoing effectual operations' announcements are gathered from \textit{CurrUpdates}). While updating aggregate values throughout its update traversal (on behalf of both ongoing effectual operations with smaller timestamp and itself), it will not create a version for each effectual operation, but rather only one version for each batch of effectual operations with timestamps $leq$ that of an ongoing aggregate query. %
This change will add $O(\textit{concUpdatingQueries})$ %
to the complexity of the aggregate query, where \textit{concUpdatingQueries} denotes the number of threads running aggregate queries whose execution interval overlaps with that of the analyzed operation (and is bounded by both \textit{concQueries} and $t$).
The second change is writing in existing versions instead of in new ones, even though this is a write with a newer timestamp, as long as no ongoing aggregate query has timestamp $\geq$ the existing one and $<$ the new one (according to the announced query timestamps), because there is no need to keep versions with timestamps that no queries will look at. This could be done using double-width \cas{} or a \cas{} of an object with the two fields - the value and the timestamp, in order for aggregate queries to obtain the value associated with the correct timestamp. %

Additionally, to make aggregate queries run faster (though without effect on their asymptotic time complexity), they may obtain left and right pointers throughout their traversals using versioned reads only without plugging in the effect of effectual operations in \textit{currUpdates}, which will save searching \textit{currUpdates}. This is fine for all nodes except for at the target area, where current effectual operations must be taken into account because they might have not yet linked a node that the query should reach, or have not yet unlinked a node it should not reach.%
Thus, an aggregate query should plug in the effect of operations in \textit{currUpdates} to the target area, either by retroactively going back up to the leaf's grandparent after reaching the leaf and re-calculating the child pointers, or by checking in advance on each node during the traversal whether the next-next child is a leaf and if so search \textit{currUpdates} for relevant operations.

\bibliographystyle{plainurl}
\bibliography{refs}

\end{document}